\documentclass[fleqn,10pt]{wlscirep}                
\usepackage{float}
\usepackage{graphicx}
\usepackage{mathptmx}
\usepackage{multirow}
\usepackage{amsthm}
\newtheorem{procedure}{Procedure}

\newtheorem{protocol}{Protocol}
\newtheorem{example*}{Example}
\newtheorem{remark}{Remark}
\newtheorem{definition}{Definition}
\newtheorem{theorem}{Theorem}
\newtheorem{lemma}{Lemma}
\usepackage[square, numbers, comma, sort&compress]{natbib}
\usepackage{physics}
\usepackage{graphicx}
\usepackage{subcaption}
\usepackage{enumitem}
\usepackage{mathtools}
\usepackage{amssymb}
\newcommand{\undertilde}[1]{\underset{\widetilde{}}{#1}}

\usepackage{siunitx}
\usepackage{pgfplots}

\newtheorem{fact}{Fact}
\usepackage{glossaries}
\newacronym{sqc}{SQC}{Semi-Quantum Communication}
\newacronym{sqdc}{SQDC}{Semi-Quantum Direct Communication}
\newacronym{sbp}{OBP}{One-Bit Protocol}
\newacronym{pns}{PNS}{Photon Number Splitting}
\newacronym{nk}{EKSQPC}{Economic Keyless Semi-Quantum Point-to-Point Communication}
\newacronym{tc}{TF}{Tele-Fetch}
\newacronym{mrad}{MRAD}{MRA Detection}
\newacronym[plural=MRAs, firstplural= Measure and Replay Attacks (MRAs)]{attackmr}{MRA}{the Measure and Replay Attack}
\newacronym[plural=qubits, firstplural=quantum bits (qubits)]{qb}{qubit}{quantum bit}
\newacronym[plural=cbits, firstplural=classical bits (cbits)]{cb}{cbit}{classical bit}
\newacronym[plural=quregisters, firstplural=quantum bit registers (quregisters)]{qbr}{quregister}{quantum bit register}
\newacronym[plural=cregisters, firstplural=classical bit registers (cregisters)]{cbr}{cregister}{classical bit register}
\newacronym{iid}{iid}{independent and identically distributed}
\newacronym{pbacc}{PBACC}{Public Bidirectional Authentic Classical Channel}
\newacronym{pbqc}{PBQC}{Public Bidirectional Quantum Channel}
\newacronym{renk}{REKSQPC}{Rate Estimation EKSQPC}
\newacronym{qkd}{QKD}{Quantum Key Distribution}
\newacronym{sqkd}{SQKD}{Semi-Quantum Key Distribution}
\newacronym[plural=PMs]{pm}{PM}{Positive MRAD}
\newacronym{ept}{EPT}{Entanglement Preservation Time}
\newacronym[plural=r.v., firstplural=Random Variables]{rv}{r.v.}{Random Variable}
\newacronym[plural=MITMs, firstplural=Man-In-The-Middle Attacks (MITMs)]{mitm}{MITM}{Man-In-The-Middle Attack}

\def\iid{\overset{\mbox{\scriptsize iid}}{\sim}}

\title{Keyless Semi-Quantum Point-to-point Communication Protocol with Low Resource Requirements}
\author[1,*]{Haoye Lu}
\author[2]{Michel Barbeau}
\author[1]{Amiya Nayak}
\affil[1]{University of Ottawa, School of Electrical Engineering and Computer Science (EECS), Ottawa, K1N 6N5, Canada}
\affil[2]{Carleton University, School of Computer Science, Ottawa, K1S 5B6, Canada}
\affil[*]{hlu044@uottawa.ca}

\begin{abstract}
Full quantum capability devices can provide secure communications, but they are challenging to make portable given the current technology. Besides, classical portable devices are unable to construct communication channels resistant to quantum computers. Hence, communication security on portable devices cannot be guaranteed. Semi-Quantum Communication (SQC) attempts to break the quandary by lowering the receiver's required quantum capability so that secure communications can be implemented on a portable device. However, all SQC protocols have low qubit efficiency and complex hardware implementations. The protocols involving quantum entanglement require linear \gls*{ept} and linear quregister size. In this paper, we propose two new keyless SQC protocols that address the aforementioned weaknesses. They are named Economic Keyless Semi-Quantum Point-to-point Communication (EKSQPC) and Rate Estimation EKSQPC (REKSQPC). They achieve theoretically constant minimal \gls*{ept} and quregister size, regardless of message length. We show that the new protocols, with low overhead, can detect Measure and Replay Attacks (MRAs). REKSQDC is tolerant to transmission impairments and environmental perturbations. The protocols are based on a new quantum message transmission operation termed Tele-Fetch. Like QKD, their strength depends on physical principles rather than mathematical complexity.
\end{abstract}

\begin{document}

\flushbottom
\maketitle
\section{Introduction}
Two full quantum capability devices can communicate securely with \gls*{qkd} \citep{PhysRevLett.85.441, PhysRevLett.67.661, PhysRevA.65.032302, Lo2050}. In this protocol, two communicants have to be armed with advanced quantum components including quantum registers, programmable quantum circuits and quantum generators. Most of them can only function under stable and well-configured environments and occupy large space. So, it is challenging to implement secure communications on portable devices. On the other hand, quantum computers can efficiently break RSA cryptosystem \citep{doi:10.1137/S0097539795293172}, the security foundation of almost all classical communication protocols. Thence, the communication security of portable devices is in imminent danger of collapse. 

\gls*{sqc} intends to break the predicament by limiting the quantum capability of the receiver without dampening the transmission security. The quantum components for realizing limited quantum capability can be designed compact, simple and robust so that they could be integrated into a portable device. The discussions start from two \gls*{sqkd} protocols reported by Boyer et al. \citep{PhysRevLett.99.140501, PhysRevA.79.032341}. Compared with \gls*{qkd}, the receiver Bob needs only to perform four quantum operations: (1) generate \glspl*{qb} in the Z-basis, (2) measure \glspl*{qb} in the Z-basis, (3) permute \glspl*{qb} and (4) access quantum channel. These two new protocols secure the communications by randomizing measurement basis and Bob's treatment on the \glspl*{qb} he receives. For concealing Bob's behaviour, reordering of the \glspl*{qb} is also required. In 2011, Jian et al. \citep{0256-307X-28-10-100301}  proposed a new \gls*{sqkd} protocol that improves \gls*{qb} efficiency (the message length with respect to the number of  \glspl*{qb} sent by Alice) from the original $12.5\%$ to roughly $50\%$ by using entangled \glspl*{qb}. But the \acrfull*{ept} for implementing the protocol is at least linear to the length of the message.  So is the \gls*{qbr} size. Li et al. \citep{li2016} showed that Bob's quantum computation task can be delegated to a third party quantum server in semi-quantum communications at the cost of a low \gls*{qb} efficiency ($6.25\%$). In 2015, Luo and Hwang \citep{Luo2016} proposed a new protocol showing that the \gls*{pbacc} is unnecessary if the two communicants have a pre-shared key. However, besides a even longer \gls*{ept} and a low \gls*{qb} efficiency ($12.5\%$), a larger \gls*{qbr} size is required for each data bit. A similar pre-shared key based protocol proposed by Almousa and Barbeau \citep{7848870} shows that Bob does not need to store any \glspl*{qb}, but the linear \gls*{ept} persists. Recently, more work concerning \gls*{sqc} is reported \citep{Shukla2017,PhysRevLett.118.220501,Wu2017,Gu2018}.

All the aforementioned protocols \citep{PhysRevLett.99.140501, PhysRevA.79.032341, 0256-307X-28-10-100301, li2016,Luo2016,7848870,Shukla2017} suffer from low \gls*{qb} efficiency. Most of them have significant large linear \gls*{qbr} size overhead and require permutation of \glspl*{qb} \citep{PhysRevLett.99.140501, PhysRevA.79.032341, 0256-307X-28-10-100301,Luo2016,7848870}. Regarding the protocol involving entangled \glspl*{qb} \citep{0256-307X-28-10-100301,Luo2016,7848870}, the quantum \gls*{ept} is at least linear. Although a six-hour record has been achieved by Zhong et al. utilizing europium ion implanted in a crystal \citep{2015Zhong}, entanglement time declines considerably should the entangled photons be propagated in an optical fiber (the most common implementation of quantum communication protocols) \citep{Inagaki:13}. Besides, involving permutations on \glspl*{qb} (not practical shortly) dooms to a low transmission efficiency and reliability. Considering that the unusual materials~(for instance, coupled electron \citep{2010Neumann} and ultracold atoms \citep{2016Dai}) are necessary for the implementation of \glspl*{qbr}, a commercial quantum network based on them is not feasible in a near future.

This paper reports a new \gls*{sqdc} protocol and a rate estimation version, named \gls*{nk} and \gls*{renk}, that address all the aforementioned issues. An innovative operation, called \gls*{tc}, utilizes entangled \gls*{qb} pairs to transmit messages. It is at the core of the \gls*{sbp}. The results of measurements on the pairs fall in a predesigned set of values because of the entanglement, but do not carry any useful information. The design makes the \gls*{sbp} functioning without a pre-shared key and fully resistant to information leakage even if the \glspl*{qb} are intercepted. Besides, the protocol uses the same quantum circuit as the one to detect \gls*{attackmr} (called \gls*{mrad}) and thus, not only saves the quantum resources but also becomes the cornerstone of the \gls*{nk} and \gls*{renk} protocols. Because Alice performs the same quantum procedures in both protocols (\gls*{sbp} or \gls*{mrad}), Bob does not need to communicate with Alice until all quantum procedures (Alice's and Bob's) are completed. Alice and Bob execute \glspl*{mrad} using a small portion of the measurement results before using the \gls*{pbacc} to translate the rest into valid messages. The protocol is proved fully secure under the assumption that \glspl*{attackmr} are always detectable. As the pivot to secure the messages is a successful detection of \glspl*{attackmr}, we show that, with only $15$ probing bits, the attack detection success rate can achieve $0.995$ (under the assumption that the adversary Eve has $0.6$ possibility to attack a \gls*{qb}). The security of the protocol is enhanced considerably if a few more probing bits are added. The \gls*{qb} efficiency asymptotically reaches $100\%$ with the message length. The implementation of the \gls*{nk} protocol has low requirements on quantum resources. In particular, the \gls*{qbr} size required by Alice is as low as one, and the required \gls*{ept} is $C+2T$~(where $C$ is the time that Alice takes to generate, send and receive the \glspl*{qb}; and $T$ is the one-way time for the \glspl*{qb} to travel between Alice and Bob). We prove that both the \gls*{qbr} size and \gls*{ept} reach the theoretical minimums.

Considering that the entanglement of \glspl*{qb} may not always persist during the transmission of \glspl*{qb}, we assume that there is a probability $\omega$ that the entanglement involving a \gls*{qb} is destroyed as it can be disturbed by hardware imperfection and environmental disturbance. Under this assumption, we design a statistical test to compare $\omega$ with the probability that a \gls*{qb} is attacked or disturbed. When a significant difference is observed, Alice concludes that Eve perpetrated attacks and aborts the execution of the protocol. Therefore, the communication is not eavesdropped successfully. Compared with the original \gls*{nk}, more probing bits are required to achieve the same detection success rate; however, the overhead is still low. In particular, our simulation results reveal that $60$ probing bits are enough to detect almost all attacks when $\omega$ is unknown. If the rate is given, then $40$ probing bits are enough to achieve the same detection success rate. 

This paper is a revised and extended version of a preliminary workshop paper \citep{8269077} in which we introduced the original \gls*{nk} protocol. Compared to the workshop paper, this paper articulates the original protocol as well as its analysis with more details. Based on this, we report an upgraded and practical version, \gls*{renk}, with its security, resource requirements and transmission efficiency analysis. Moreover, we also provide a detailed comparison with other typical \gls*{sqkd} and \gls*{sqdc}~protocols.

The rest of the paper is organized as follows. In Section~\ref{secNoteBack}, we review Bell measurement and \gls*{mrad}, which are integrated in the new protocol. In Sections~\ref{secNP} and \ref{secRev}, we introduce our new protocols including a rate estimation version taking into account the probability that a \gls*{qb} is disturbed. We also do a security analysis and discuss simulation results. In Section~\ref{secEvaluation}, we talk about their quantum resource requirements and transmission overhead. Finally, we draw the conclusions in Section~\ref{secConclusion}.

\section{Background}\label{secNoteBack}
The section starts from a brief review of EPR pairs states and Bell measurement on which our new protocol heavily relies. Then we introduce the \glspl*{attackmr} as well as its detection algorithm (\gls*{mrad}) that secures the data transmission of the new protocol.

In this paper, \glspl*{cb} are denoted by lowercase English letters, and a \gls*{cb} sequence is represented by an uppercase English letter over a tilde. For instance, $\undertilde{M}= m_1m_2\cdots m_t$ is a \gls*{cb} string of length $t$. Qubits are denoted by Greek letters and Bell states by Bold English capital letters.

\subsection{EPR pairs and Bell measurement}\label{secBSBM}
A pair of \glspl*{qb} that are together in Bell state is called an EPR pair. Bell states have four types: $\ket{\Phi^+}  = \frac{1}{{\sqrt{2}}}\cdot(\ket{00}+\ket{11})$, $\ket{\Phi^-} = \frac{1}{{\sqrt{2}}}\cdot{(\ket{00}-\ket{11})}$, $\ket{\Psi^+} = \frac{1}{{\sqrt{2}}}\cdot(\ket{01}+\ket{10})$, and $\ket{\Psi^-} = \frac{1}{{\sqrt{2}}}\cdot(\ket{01}-\ket{10})$; we can use the Bell measurement (B.M.) (Figure~\ref{figbm}) to identify them.  The inputs of the circuit are two \glspl*{qb} $\gamma_A$ and $\gamma_B$, and the outputs are two \glspl*{cb} $e_1$ and $e_2$.

\begin{figure}[!h]
\centering
\includegraphics[width=0.3\textwidth]{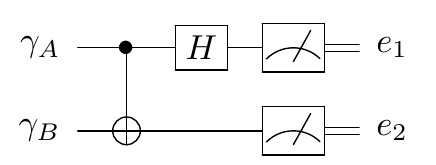}
\caption{The Bell measurement Circuit.}\label{figbm}
\end{figure}
If $\gamma_A\gamma_B$ is an EPR Pair (Bell state), then the outputs are deterministic and listed in Table~\ref{tableBMBS}; otherwise, $\gamma_A\gamma_B$ is mapped into a Bell state stochastically. For $\gamma_A\gamma_B$ equal to $\ket{00}$ , $\ket{01}$, $\ket{10}$ or $\ket{11}$, the distributions of the outputs are listed in Table~\ref{tableBMNBS}.

\subsection{Measure and Replay Attack and the method of detection}\label{secMrad}

A simplified version of \glspl*{mitm} is called replay attack. The attacker Eve deceives the truthful listener(s) by replaying messages outside the expected context so that the listener believes that the protocol has been executed successfully~\citep{Malladi02onpreventing}.

We can perpetrate a similar attack in the context of quantum communications. Assume that  \glspl*{qb} sent or received by Alice and Bob can be intercepted by Eve. As a \gls*{qb} is sent to Bob by Alice through the quantum channel, Eve uses the Z-basis to measure it. If the measurement result is zero, then Eve sends $\ket{0}$ to Bob; otherwise, she sends $\ket{1}$. As the replay of the message follows the measurement, we call the attack  \acrfull*{attackmr} \citep{7848870}.

The following method demonstrates how to utilize the EPR pairs and Bell measurement to detect \glspl*{attackmr}. The Luo and Hwang's protocol \citep{Luo2016} and Almousa and Barbeau's protocol \citep{7848870} also apply  similar ideas for the attack detection. 

\textbf{\acrfull*{mrad}}:
\begin{enumerate}[label=D\arabic*]
\item\label{mradAlicePrep} Alice randomly picks a \gls*{cb} $i=0$ or $1$, based on which she generates an EPR pair $\mathbb{E}$ (if $i = 0$, $\mathbb{E}=\ket{\Phi^+}$; else, $\mathbb{E}=\ket{\Psi^-}$).
\item\label{mradAliceSend} $\mathbb{E}$ consists of two \glspl*{qb}, $\gamma_A$ and $\gamma_B$. Alice keeps $\gamma_A$ and sends $\gamma_B$ (named the probing bit) to Bob. 
\item\label{mradBobRef}	$\gamma_B$ is reflected by Bob to Alice.
\item\label{mradAliceRec} After receiving $\gamma_B'$, Alice applies Bell measurement (Figure~\ref{figbm}) on $\gamma_A\gamma_B'$ to obtain $e_1$ and $e_2$.
\item\label{mradCorrect} The combination of $e_1$ and $e_2$ indicates the EPR pair that the circuit measured. We consider the protocol secure (denoted by zero) if the measured EPR pair agrees with the one Alice produced in Step~\ref{mradAlicePrep}. If not, an \gls*{attackmr} is detected (denoted by one).
\end{enumerate}

\begin{example*}
Suppose Alice and Bob implement \gls*{mrad} for \gls*{attackmr} detection. Without loss of generality, assume Alice picks $i = 1$ and thus produces a corresponding EPR pair $\mathbb{E}=\ket{\Psi^-}=\frac{1}{{\sqrt{2}}}\cdot(\ket{01}-\ket{10})=\gamma_A\gamma_B$. Alice intends to send $\gamma_B$ to Bob; however, Eve intercepts and measures it and gets the measurement result $r = 0$. Simultaneously, $\gamma_A$ retained by Alice collapses to $\ket{1}$ due to the entanglement. Eve produces a new \gls*{qb} ($\ket{0}$) correspondingly and send it to Bob. Bob does nothing but reflects it back to Alice. $\gamma_A$ and the received $\ket{0}$ are paired together and measured by Alice using the Bell measurement circuit. Notice that $\gamma_A$ has collapsed to $\ket{1}$. So the \gls*{qb} pair measured by the circuit is $\gamma_A\gamma_B = \ket{10}$. By Table~\ref{tableBMNBS}, we have $50$ percent possibility to get $e_1e_2 = 01$ (and so deduce that the input is $\ket{\Psi^+}$ by Table~\ref{tableBMBS}, a true positive) and to get $e_1e_2 = 11$ (and thus deduce that the input is $\ket{\Psi^-}$, a false negative). 
\end{example*}

The example shows that if Alice picks $i = 1$ and the measurement result $r$ of Eve is zero, there is $50$ percent possibility for Alice to deduce that the protocol is secure although the attack is perpetrated.  By Table~\ref{tableBMNBS}, we can draw the same conclusion for any choice of $i$ and $r$. Hence, Lemma~\ref{lemmaDectmra} and Theorem~\ref{umdn} follow. 
\begin{lemma}\label{lemmaDectmra}
	Provided that Eve attacks the probing bit, there is $0.5$ possibility for \gls*{mrad} to detect an \gls*{attackmr}. 
\end{lemma}

\begin{theorem}\label{umdn}
	If \gls*{mrad} are repeated $n$ times, we have $1 - 0.5^n$ probability to detect \glspl*{attackmr} given that Eve attacks $n$ probing bits. 
\end{theorem}
\begin{proof}
		$Pr[\mbox{detect \glspl*{attackmr}}] = 1 - Pr[\mbox{\gls*{mrad} fails}]^n = 1 - 0.5^n$
\end{proof}
\begin{remark}
	\gls*{mrad} essentially checks whether the probing bits sent by Alice had been measured by anybody else, but it cannot tell who measured them. It can detect \glspl*{attackmr} only because Alice knows that Bob does not measure probing bits. Therefore, if  any measurement is detected, it must be due to an attack. 
\end{remark}

Because, the operations defined in Steps~\ref{mradAlicePrep} and \ref{mradCorrect} are applied again in the sequel, we define them formally as follows,
\begin{definition}[Generating corresponding EPR pairs ($F$), Step~\ref{mradAlicePrep}]\label{funcMBS}
	Function $F$ maps a \gls*{cb} to an EPR pair such that $0 \mapsto \ket{\Phi^+}$, and $1 \mapsto \ket{\Psi^-}$
\end{definition}

\begin{definition}[Alice Examines (AE), Step~\ref{mradCorrect}]\label{ae}
	The function $AE(e_1, e_2, i): \{0,1\}^3 \rightarrow \{0,1\}$ equals zero if $e_1=e_2=i$; otherwise, it equals one. Recall that zero and one indicate negative and positive detection results, respectively. 
\end{definition}

\section{New Protocol}\label{secNP}
In this section, we propose a new \gls*{sqdc} protocol called \acrfull*{nk}. We start the discussion with an introduction to a data transmission protocol called \gls*{sbp}, which is a building block of \gls*{nk}  (not self-contained). Assuming that there are no \glspl*{attackmr}, we show that \gls*{sbp} is secure (Theorem~\ref{thmSbpSec}). In the design of \gls*{sbp}, Bell measurement seems redundant. It is intended for sharing the quantum circuit with \gls*{mrad} (Remark~\ref{rmkSBPDesign}).  The considerable benefits of this design are discussed in Section~\ref{secEvaluation}. To meet the assumption of Theorem~\ref{thmSbpSec}, we integrate \gls*{mrad} and \gls*{sbp} to get \gls*{nk}. If we assume that \gls*{nk} detects all \glspl*{attackmr}, then it is provably secure (Theorem~\ref{thmNkSecure}). 

\subsection{\acrfull*{sbp}}\label{sectionsbp}
\begin{protocol}[\gls*{sbp}]
A one-bit message $m$ (zero or one) is sent to Bob by Alice. We need a \gls*{pbacc} as well as a \gls*{pbqc}. The protocol functions as follows:
	\begin{enumerate}[label=P\arabic*]
		\item\label{SimAliceSend} Alice randomly picks a \gls*{cb} $i$  and  generates a corresponding EPR pair $\mathbb{E} = F(i)$ (Definition~\ref{funcMBS}) consisting of two \glspl*{qb} (denoted by $\gamma_A$ and $\gamma_B$).
		\item Alice keeps $\gamma_A$ and sends $\gamma_B$ to Bob.
		\item\label{SimBobRec} Upon reception, the \gls*{qb} $\gamma_B$ is measured by Bob in the Z-basis with the measurement result $u_B$ (simultaneously, $\gamma_A$ collapses because of the entanglement with $\gamma_B$). At the same moment, Bob sends a pre-prepared \gls*{qb} $\gamma_B^* = \ket{0}$ to Alice and informs her that he has measured $\gamma_B$ via the \gls*{pbacc}.
		\item\label{SimAliceMeas}Alice pairs $\gamma_A$ (retained in Step~\ref{SimAliceSend}) with $\gamma_B'$ and performs a Bell measurement on $\gamma_A\gamma_B^* = \gamma_A\ket{0}$ to get $e_1$ and $e_2$.
		
		\item\label{SimAliceProc}According to Table~\ref{tableBMNBS}, $e_1e_2=00 \mbox{ or } 10$ implies that $\gamma_A=\ket{0}$, and $e_1e_2=01 \mbox{ or } 11$ indicates that $\gamma_A=\ket{1}$. Combining $\gamma_A$ with the EPR pair $\mathbb{E}$ Alice selected (recorded by $i$) in Step~\ref{SimAliceSend}, Alice learns the measurement result $u_B$ of Bob in Step~\ref{SimBobRec}.
		In particular, if $i=0$, then the EPR pair she generated was $\ket{\Phi^+}$. Then $\gamma_A = \ket{0}$ implies $u_B = 0$, and $\gamma_A =\ket{1}$ implies $u_B = 1$. Similarly, if $i = 1$, the EPR pair that Alice generated was $\ket{\Psi^-}$. Then if $\gamma_A = \ket{0}$, $u_B = 1$; else, $u_B = 0$. 
		\item\label{AliceConfirm} Provided that $u_B = m$, Alice informs Bob, via the \gls*{pbacc}, that $u_B$ is the correct value. Otherwise, she informs Bob to take $1-u_B$.
	\end{enumerate}
\end{protocol}

\begin{remark}\label{rmkSBPDesign}
	The pre-generated \gls*{qb} $\gamma_B^* = \ket{0}$ in Step~\ref{SimBobRec} is unnecessary to implement \gls*{sbp}. So is the Bell measurement in Step~\ref{SimAliceMeas}. In fact, in Step~\ref{SimBobRec}, Bob only needs to notify Alice that he has measured $\gamma_B$, and, in Step~\ref{SimAliceMeas}, Alice simply uses the Z-basis to get the value of $\gamma_A$. Here, we intendedly implement \gls*{sbp} with redundant operations so that the new protocol (\gls*{nk}, introduced in Section~\ref{nk}) can use a single quantum circuit to implement both the attack detection (\gls*{mrad}) and data transmission (\gls*{sbp}) protocols. We discuss the design and its benefits in details in Section~\ref{nk}, and more performance analysis is conducted in Section~\ref{secEvaluation}. 
\end{remark}

The actions specified in Steps~\ref{SimAliceProc} and \ref{AliceConfirm} are used subsequently. We define them formally as follows. In Step~\ref{SimAliceProc}, Alice learns $r_B$ held by Bob with no contact. So the function is called \acrlong*{tc}. 
\begin{definition}[\acrlong*{tc}]\label{teleFetch}
	With the parameters $e_1$, $e_2$ and $i$ (Step~\ref{SimAliceSend}), function \acrlong*{tc} $\gls*{tc}(e_1, e_2, i)$ returns the value of $r_B$ (zero or one) based on the rule contained in Step~\ref{SimAliceProc}.
\end{definition}

Besides, in Step~\ref{AliceConfirm}, Alice rectifies the measurement result $r_B$ of Bob. As a result, we call the procedure Rectify.

\begin{procedure}[Rectify]\label{OCorrect}
	Based on the single bit message $m$ and value of $u_B$ (acquired in Step~\ref{SimAliceProc}), Alice informs Bob to apply the proper operation on $u_B$ by sending either the signal KEEP or FLIP via the \gls*{pbacc}. If KEEP is received, Bob considers $u_B$ as the message Alice sends; if not, he takes $1-u_B$.
\end{procedure}

The next theorem discusses the security of \gls*{sbp}. 

\begin{theorem}\label{thmSbpSec}
	As long as Bob gets the \gls*{qb} $\gamma_B$ sent by Alice without \gls*{attackmr}, \gls*{sbp} is secure. 
\end{theorem}
\begin{proof}
By assuming the absence of \gls*{attackmr}, we essentially assume that Steps~\ref{SimAliceSend} to \ref{SimBobRec} are secure (only a confirmation is sent by Bob in Step~\ref{SimBobRec}). No communication happens in Step~\ref{SimAliceMeas} and \ref{SimAliceProc}. The last step involves a message sent by Alice which is irrelevant to the one-bit message $m$. So Step~\ref{AliceConfirm} is sheltered, too. Thence, to sum up, \gls*{sbp} is secure.
\end{proof}

\begin{remark}\label{rmkAuth}
An authentic classical channel is the prerequisite for the security of \gls*{sbp}. In the communication of Alice and Bob, it is significant to verify their identities and to ensure that their unencrypted messages are not altered. Namely, they should be resistant to \glspl*{mitm}.
\end{remark}

Theorem~\ref{thmSbpSec} shows that only when there is no \gls*{attackmr}, \gls*{sbp} is secure. However, \gls*{sbp} has no capability to detect \glspl*{attackmr}. Notice that \gls*{mrad} can detect \glspl*{attackmr} and thus can secure the data transmission of \gls*{sbp} by Theorem~\ref{thmSbpSec}. If we combine \gls*{sbp} and \gls*{mrad} together, we get the protocol discussed in the following subsection.

\subsection{\acrfull*{nk}}\label{nk}
We implement the protocol \gls*{nk} over the hardware of \gls*{sbp}. Specifically, there are a \gls*{pbacc} and a \gls*{pbqc} linking Alice and Bob. The following four procedures contain all the activities demanding quantum resources in \gls*{nk}.

\begin{procedure}[Alice sends]\label{OAS}
	In the $k^{th}$ transmission of Alice, she picks a random \gls*{cb} $i_k$ and stores it in a classical register. Then she produces an EPR pair $F(i_k)$, keeps the first \gls*{qb} $\gamma_{kA}$ and transmits the second  \gls*{qb} $\gamma_{kB}$ to Bob.
\end{procedure}

\begin{procedure}[Bob measures]\label{OBM}
	After receiving the $k^{th}$ \gls*{qb} from Alice, Bob measures it in the $Z$-basis and gets the result $u_k$. At the same moment, a pre-prepared \gls*{qb} $\ket{0}$ is sent back to Alice. Furthermore, Bob takes the record that he measured the $k^{th}$ \gls*{qb}. 
\end{procedure}

\begin{procedure}[Bob reflects]\label{OBR}
	The $k^{th}$ \gls*{qb} from Alice is reflected back without measurement by Bob. Also, he takes the record that he reflected the $k^{th}$ \gls*{qb} he received.
\end{procedure}

\begin{procedure}[Alice measures]\label{OAM}
	Alice receives the $k^{th}$ \gls*{qb} $\gamma_{kB}^*$ and performs Bell measurement on $\ket{\gamma_{kA}\gamma_{kB}^*}$ ($\gamma_{kA}$ was retained by Alice in Step~\ref{NKAliceSend} while implementing Procedure~\ref{OAS}) and records the measurement result as $e_{1k}e_{2k}$.
\end{procedure}

\begin{protocol}[\gls*{nk}]
	Assume that a message $\undertilde{M} = m_1m_2\cdots m_s$ of length $s$ is sent to Bob by Alice, and extra $r$ bits are added to detect \glspl*{attackmr}. Then the protocol functions as follows:
	\begin{enumerate}[label=C\arabic*]
		\item\label{NKAliceSend} Alice runs Procedure~\ref{OAS} for $s+r$ times and records the values of $i_k$ in string $\undertilde{I} = i_1 i_2 \cdots i_{s+r}$.	
		
		\item \label{NKBobMR} Bob randomly selects $s$ \glspl*{qb} (data bits) from the $s+r$ \glspl*{qb} that Alice sends to implement Procedure~\ref{OBM}. Regarding the residual $r$ \glspl*{qb} (probing bits), he executes Procedure~\ref{OBR}. All the measurement results $u_k$ from Procedure~\ref{OBM} are recorded in a new string $\undertilde{U} = u_1 u_2 \cdots u_s$ (after reindexing but preserving the order).
		\item\label{NKAliceMeas}  Alice performs Procedure~\ref{OAM} on the $s+r$ \glspl*{qb} that Bob sends back, and records the measurement results $e_{1k}e_{2k}$ in two strings $\undertilde{E_1} = e_{11}e_{12} \cdots e_{1(s+r)}$ and $\undertilde{E_2} = e_{21}e_{22} \cdots e_{2(s+r)}$, respectively.
		\item\label{NKBobNotify} Bob sends a binary string $\undertilde{P} = p_1 p_2 \cdots p_{s+r}$ to Alice through the \gls*{pbacc} to inform her about which \glspl*{qb} were reflected or measured in Step~\ref{NKBobMR}. For $k = 1, 2, \cdots, s+r$, $p_k = 0$ indicates that Bob reflected the $k^{th}$ \gls*{qb}, and $p_k = 0$ represents he measured it. 
		\item\label{NKAliceCheck} Alice iterates through  $\undertilde{P}$ sent by Bob. For $k = 1, 2, \cdots, s+r$, when $p_k = 0$, Alice applies function $AE(e_{1k},e_{2k},i_k)$ in Definition~\ref{ae}. If $AE(e_{1k},e_{2k},i_k) = 1$, then the $k^{th}$ \gls*{qb} sent by Alice is attacked by Eve (\gls*{attackmr}). Then the protocol is insecure and terminated. While if $p_k = 1$, Alice evaluates function $\gls*{tc}(e_{1k}, e_{2k}, i_k)$ in Definition~\ref{teleFetch} and records the value $c_{k}$. Remark that, before reindexing, $c_k$ coincides with $u_k$ in Step~\ref{NKBobMR}.
		\item\label{NKAliceFetch} Since $s$ \glspl*{qb} are measured by Bob, Alice applies function $\gls*{tc}$ $s$ times in Step~\ref{NKAliceCheck}. She records the values of $c_k$ in $\undertilde{C} = c_1 c_2 \cdots c_s$ (after reindexing without altering the order). Note that $\undertilde{C}$ coincides with $\undertilde{U}$ that is owned by Bob.
		\item\label{NKCorrect} Alice and Bob execute Procedure~\ref{OCorrect} with patemeters $m_k$, $c_k$ and $u_k$ ($k = 1, 2, \cdots, s$). Then Bob receives the message sent by Alice. 
	\end{enumerate}
\end{protocol}

Remarkably, Alice and Bob implement Steps~\ref{NKAliceSend} to \ref{NKAliceMeas} in parallel rather than sequentially. As a result, Alice is only required to be equipped with a small and fixed number of \glspl*{qbr}. Also, the time for Alice to keep the entanglement is a small constant irrelevant to the message length (we elaborate on this highlight in Section~\ref{secQuanRec}). The protocol essentially distributes a random string of length $m$ between Alice and Bob, which implies that our protocol is also a \gls*{sqkd} protocol. After sharing a binary string, Bob can receive messages from Alice by implementing Procedure~\ref{OCorrect}. To mitigate the cost of sharing keys, Step~\ref{NKCorrect} can be repeated for several message transmissions before adopting a new shared string $\undertilde{U}$ ($=\undertilde{C}$) (by repeating Steps~\ref{NKAliceSend} to \ref{NKAliceFetch}).

The \gls*{nk} protocol is an integration of \gls*{sbp} and \gls*{mrad}. Figure~\ref{figRelation} demonstrates that, in the first four steps, the operations belonging to Alice coincide. Notice that these four steps include all the operations of \gls*{sbp} and \gls*{mrad} that require quantum resources. As a result, without knowledge of the protocol she is in fact executing, Alice can use one quantum circuit to accomplish all the quantum operations required by either of the protocols.  

In Figure~\ref{figRelation}, we juxtapose the first four steps of \gls*{sbp} and \gls*{mrad} marked with the step numbers used to present them. On the right, the step numbers used in the \gls*{nk} protocol are also provided.

\begin{figure}[!h]
	\centering
	\includegraphics[width=0.4\textwidth]{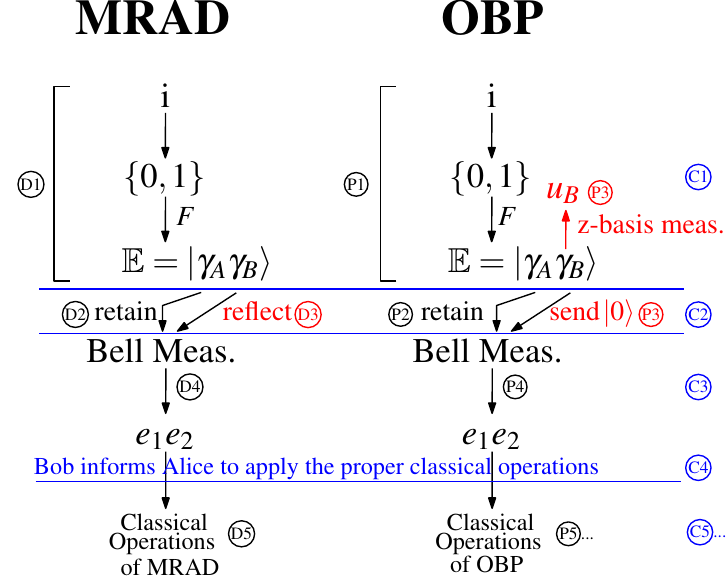}
	\caption{The relationships between \gls*{nk}, \gls*{sbp} and \gls*{mrad}. The actions in red are made by Bob and those in Black are performed by Alice. The step numbers in blue are the corresponding actions in \gls*{nk}. In \gls*{sbp} and \gls*{mrad}, quantum operations (first four steps) are quite similar except Bob's treatment on the \glspl*{qb} sent by Alice. 
	}\label{figRelation}
\end{figure}

From Figure~\ref{figRelation}, only the operation made by Bob differentiates \gls*{sbp} from \gls*{mrad}. Namely, Bob decides which protocol is being implemented. Specifically, to decide the protocol being applied,  Bob either measures $\gamma_{kB}$ (and send a pre-prepared $\ket{0}$ simultaneously) or reflects it. Note that the pre-preparation of the $\ket{0}$, instead of generating it on demand, secures \gls*{nk} against the delay and reflection attacks \citep{7848870}. A reflected $\gamma_{kB}$ functions as a probing bit to detect \glspl*{attackmr} (then Alice and Bob implement \gls*{mrad}), and a measured $\gamma_{kB}$ works a data bit for data exchange (then Alice and Bob implement \gls*{sbp}). After completing the first three steps of \gls*{nk}, Bob informs Alice of the qubits reflected or measured by sending a notification through the \gls*{pbacc}. Based on the message, Alice applies the corresponding classical operations to complete \gls*{mrad}s or \gls*{sbp}s.

\begin{remark}\label{rmksmn}
	The protocol being implemented is determined by Bob. If Bob chooses Measure, then it is \gls*{sbp}. If he chooses Reflect, then it is \gls*{mrad}. In \gls*{nk}, Bob selects Measure $s$ times and Reflect $r$ times. So Alice and Bob execute \gls*{sbp} $s$ times and \gls*{mrad} $r$ times. 
\end{remark}

\subsection{Security analysis of \gls*{nk}}\label{secSecurity}
The \gls*{nk} inherits the security of \gls*{sbp} and functions under the same assumption -- Alice and Bob must be connected by an authentic classical channel (Remark~\ref{rmkAuth}). By Remark~\ref{rmksmn}, \gls*{nk} with $s$ data bits and $r$ probing bits is equivalent to $s$ \gls*{sbp}s and $r$ \gls*{mrad}s. Recall that \gls*{mrad} is for detecting \glspl*{attackmr} and thus secures \gls*{sbp} (Theorem~\ref{thmSbpSec}). When \gls*{mrad}s are performed $n$ times, the possibility of detecting \glspl*{attackmr} is $1-0.5^n$ (Theorem~\ref{umdn}). In particular, since there are $r$ times executions of \gls*{mrad}s in the \gls*{nk} protocol, we have $1-0.5^r$ success rate of detection given that all \glspl*{qb} sent by Alice are measured by Eve. If we generalize the problem by assuming that Eve perpetrates \glspl*{attackmr} on the \glspl*{qb} with a fixed probability, we have the theorem as follows. 	
\begin{theorem}\label{thmDetectMRAp}
	Suppose that Alice and Bob implement the \gls*{nk} protocol with $s$ data bits and $r$ probing bits. For each \gls*{qb} sent by Alice, if Eve has probability $p$ to perpetrate \glspl*{attackmr}, then Alice has the probability $1-\left(1-p/2\right)^r$ to detect it. 
\end{theorem}
\begin{proof}
	Let $A$ be the number of probing bits attacked by Eve. As Eve has possibility $p$ to   perpetrate an \gls*{attackmr} on each probing bit, $A$ follows a binomial distribution having success rate $p$ with $r$ trials. Let $D$ be a boolean \gls*{rv} such that $D = 1$ if Alice detects an attack and $D = 0$ if not. Then the expectation $E(D)$ is the probability of detecting an attack, and it satisfies
	
	{\small$\begin{aligned}
		E[D] &= 0\cdot Pr[D = 0] + 1\cdot Pr[D =1] = Pr[D=1] = \sum\limits_{h=0}^{r} Pr[D=1|A=h] Pr[A=h] && \mbox{(Law of total probability)}\\
		& = \sum\limits_{h=0}^{r} \left(1-0.5^h\right) \binom{r}{h}p^h(1-p)^{r-h} && \mbox{(by Theorem~\ref{umdn})}\\
		& = \sum\limits_{h=0}^{r}\binom{r}{h}p^h(1-p)^{r-h} - \sum\limits_{h=0}^{r}0.5^{h}\binom{r}{h}p^h(1-p)^{r-h} = 1-\left(0.5 p+(1-p)\right)^r && (\mbox{Binomial expansion})\\
		& = 1-\left(1-0.5 p\right)^r
		\end{aligned}$}\vspace{-1.8em}
		
\end{proof}

Regarding Theorem~\ref{thmDetectMRAp}, if $p=1$, Eve attacks all the \glspl*{qb} that Alice sends. Then the probability of detecting an \gls*{attackmr} is $E(D) = 1 - (1-0.5)^r = 1 - 0.5^r$, which is consistent with the discussion at the beginning of Section~\ref{secSecurity}.

Figure~\ref{figEveP} plots the trend of the detection success rate, calculated according to the formula stated in Theorem~\ref{thmDetectMRAp}. We also scatter the experimental results (the points) from simulation. According to what the legend shows, the points and curves colored the same  share the same attack probability $p$. The results of the simulation agree with the theoretical analysis in Theorem~\ref{thmDetectMRAp}.
The figure shows that the detection success rate  approaches to one more rapidly as $p$ increases. This trend is due to the fact that a higher attack rate leads to a higher average number of affected probing bits and thus boosts the detection success rate. A similar trend can be observed if the probing bit number $r$ increases.

\begin{figure}[!h]
	\centering
	\includegraphics[width=0.65\columnwidth]{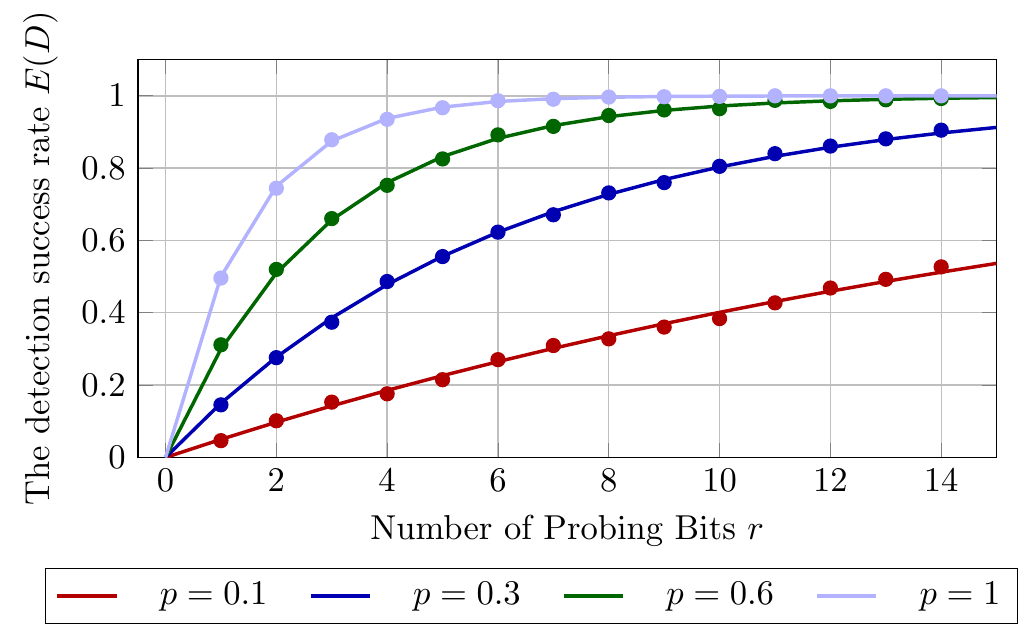}
	\caption{The detection success rate $E(D)$ with respect to the probing bit number $r$ by selecting various attack probability~$p$.}\label{figEveP}
\end{figure}

The next theorem shows the \gls*{nk} protocol is secure under the assumption that we can always detect \glspl*{attackmr}.
\begin{theorem}\label{thmNkSecure}
	The \gls*{nk} protocol is secure if \glspl*{attackmr} can always be detected. 
\end{theorem}
\begin{proof}
	Alice and Bob terminate the protocol if an \gls*{attackmr} is detected. So the security of the message is guaranteed. Otherwise, there is no attack because of the assumption. According to Remark~\ref{rmksmn}, the \gls*{nk} protocol with $s$ data bits performs \gls*{sbp} for $s$ times. Combining with Theorem~\ref{thmSbpSec}, we conclude that the \gls*{nk} protocol is resistant to any network attack.
\end{proof}

\begin{remark}
	It is the prerequisite for Theorem~\ref{thmNkSecure} that only one \gls*{qb} is involved when Alice and Bob send, measure or reflect \glspl*{qb}. This implies that the implementation of the protocol requires a generator of an individual photon stream which is, however, currently not available.  In practice, if we use weak laser pulses out of expediency, more than one photon may be included. This enables \gls*{pns} attacks which cannot be handled by our protocol. To avoid the attacks related to \gls*{pns}, readers may refer to \citep{PhysRevLett.92.057901,PhysRevLett.94.230504,Kalashnikov:11,Zhou:14}. The same comment also applies to Theorem~\ref{thmNoAErr}.
\end{remark}

\section{\acrfull*{renk}}\label{secRev}
\gls*{nk} detects \glspl*{attackmr} and is secure assuming no hardware fault nor environmental disturbance that destroy entanglement.  So far, we ignored them for the sake of simplicity. They do exist in practice. Ignoring them produces false positives and incorrect protocol terminations. In this section, we enhance the detection part of the protocol to fix this issue. Destructions of entanglement involving probing bits may result in \glspl*{pm} whose probability is denoted by $\rho$ and estimated by its rate $$\hat{\rho}:=\frac{\mbox{Number of \glspl*{pm}}}{\mbox{Number of probing bits}}.$$ 
The destructions have two types. In particular, we say that a \gls*{qb} is disturbed if the entanglement involving it is destroyed due to a hardware imperfection or an environmental disturbance. If the destruction of the entanglement is caused by an eavesdropper Eve, we say the \gls*{qb} is attacked. We show that two times $\hat{\rho}$ is an estimator $\hat{\kappa}$ of the probability $\kappa$ that a probing bit is disturbed or attacked. Let $\omega$ denote the probability that a \gls*{qb} is disturbed. If $\omega$ is unknown, we can estimate it ahead of the protocol execution assuming that Eve does not perpetrate attacks. As no \glspl*{qb} are attacked during the estimation, $\kappa$ is reduced to $\omega$. Correspondingly, $\hat{\kappa}$ is reduced to $\hat{\omega}$, an estimator of $\omega$. During the execution of the protocol, the attacks perpetrated by Eve increase $\kappa$ and cause its deviation from $\omega$. By monitoring the difference between $\kappa$ and $\omega$, we gauge the existence of attacks and thus the security of the protocol. We use the following symbols and facts for the statistical analysis in the sequel. Let $B(n,p)$ be a binomial distribution with $n\in \mathbb{N}$ trials and success rate $p \in [0,1]$, $N(\mu, \sigma^2)$ a normal distribution with mean $\mu\in \mathbb{R}$ and variance $\sigma^2$ and $\bar{X}$ the arithmetic mean of $X$.

\begin{remark}
	We call $B(1,p)$ a Bernoulli distribution with the success rate $p$.
\end{remark}
\begin{remark}\label{rmkBinSum}
	\glspl*{rv} of binomial distributions can be added if they have the same success rate. In particular, if $X\sim B(n, p)$ and $Y\sim B(m,p)$, then $X+Y \sim B(n+m,p)$ \citep{hogg2014probability}.
\end{remark}

\begin{fact}\label{factStdNormal}
	Suppose $X\sim N(\mu_X,\sigma_X^2)$. Then $\frac{X-\mu_X}{\sigma_X}$ follows a standard normal distribution. Namely, 
	$\frac{X-\mu_X}{\sigma_X}\sim N(0,1)$. 
\end{fact}

\begin{fact}\label{factNormalFact}
	Suppose $X\sim N(\mu_X,\sigma_X^2)$ and $n \in \mathbb{R}^+$. Then, $\frac{X}{n} \sim N(\frac{\mu_X}{n},\frac{\sigma_X^2}{n^2})$. 
\end{fact}

\begin{fact}\label{factNormalAdd}
	Suppose $X\sim N(\mu_X,\sigma_X^2)$ and $Y\sim N(\mu_Y, \sigma_Y^2)$ are independent. Then, ${X-Y \sim N(\mu_X-\mu_Y,\sigma_X^2+\sigma_Y^2)}$. 
\end{fact}

Theorem~\ref{thmKappaEst} discusses the random processes in the detection of disturbed and attacked \glspl*{qb}.
\begin{theorem}\label{thmKappaEst}
	Suppose that in the \gls*{nk} protocol, Bob reflects $r$ \glspl*{qb}. Let $D_i\in\{0,1\}$ denote a \gls*{rv} of the detection result $d_i$ of the $i^{th}$ \gls*{mrad} such that:
		\[   
		d_i = 
		     \begin{cases}
		       1 &\quad\text{if the $i^{th}$ MRAD has a positive detections}\\
		       0 &\quad\text{otherwise}\\
		     \end{cases}
		\]
	Then $D_i$'s are \gls*{iid} $B(1, \rho)$. Or in short, $D_i \iid B(1, \rho)$ . The number of \glspl*{pm} (denoted by $C_\rho$) is $\sum_{i=1}^{r}D_i$, which is a binomial distribution $B(r, \rho)$. Moreover, $\rho=\frac{\kappa}{2}$. 
\end{theorem}

\begin{proof}
	Since $\rho$ is the probability of positive detection and all \glspl*{mrad} are mutually independent, $D_i \iid B(1, \rho)$ for $i = 1 \cdots r$. Then the number of \glspl*{pm} $C_\rho = \sum_{i = 1}^{n} D_i$. By Remark~\ref{rmkBinSum}, we have $C_\rho \sim B(r,\rho)$. Let $A_i$ be a \gls*{rv} such that, if the probing bit of $i^{th}$ \gls*{mrad} is disturbed or attacked, then $A_i = 1$; otherwise, $A_i = 0$. So we have, $Pr[A_i = 0] = 1- \kappa$ and $Pr[A_i = 1] = \kappa$. 
	 According to Lemma~\ref{lemmaDectmra}, if the probing bit is disturbed or attacked, the probability of a positive detection is $Pr[D_i = 1| A_i = 1] = \frac{1}{2}$. Otherwise, the probing bit is intact which implies that the result must be negative. Namely, $Pr[D_i = 1| A_i = 0]=0$. By Law of total probability, 
	 $$\rho = Pr[D_i = 1] = Pr[D_i = 1|A_i = 1]\cdot Pr[A_i = 1] + Pr[D_i = 1|A_i = 0]\cdot Pr[A_i = 0] = \frac{1}{2}\cdot \kappa + 0 \cdot (1-\kappa) = \frac{\kappa}{2}\mbox{.}$$
\end{proof}

\begin{remark}
	A binomial distribution $N(n,p)$	 has mean $np$ and variance $np(1-p)$. So the binomial distribution $C_\rho \sim B(r,\rho)$ in Theorem~\ref{thmKappaEst} has mean $r\rho= \frac{\kappa}{2} r$ and variance ${r\rho(1-\rho)=\frac{1}{2}r\kappa(1-\frac{\kappa}{2})}$ \citep{hogg2014probability}.
\end{remark}

\begin{remark}\label{rmkBinAppr}
	A binomial distribution $B(n,p)$ can be approximated by a normal distribution with the same mean and variance if $n\geq max\left\{\frac{45(1-2p)^2}{p(1-p)}, \frac{14|1-6p(1-p)|}{p(1-p)
	} \right\}$ \citep{DasGupta2010}. 
	Therefore, the binomial distribution $C_\rho \sim B(r,\frac{\kappa}{2})$ in Theorem~\ref{thmKappaEst} has a normal approximation $N(\frac{\kappa}{2} r, \frac{1}{2}r\kappa(1-\frac{\kappa}{2}))$ if $r\geq max\left\{\frac{180(1-\kappa)^2}{p(2-\kappa)}, \frac{56|1-3\kappa(1-0.5\kappa)|}{\kappa(2-\kappa)} \right\}$.
\end{remark}

Theorem~\ref{thmBern} provides a method to estimate the parameter $p$ of a Bernoulli distribution \citep{hogg2014probability}.

\begin{theorem}\label{thmBern}
\begin{sloppypar}
	Suppose that $X_i\iid B(1,p)$ for $i\in\{1,2, \cdots, n\}$. Then ${{\hat{p} =\bar{X}=\frac{\Sigma_{i = 1}^{n}x_i}{n}}}$, an unbiased estimator of $p$. 
\end{sloppypar}
\end{theorem}
\begin{sloppypar}
By Theorem~\ref{thmBern}, $\rho = \frac{\kappa}{2}$ has an unbiased estimator $\hat{\rho}={\widehat{\kappa /2} = \frac{\sum_{i=1}^r D_i}{r} = \frac{C_\rho}{r}}$. Therefore, $\kappa$ can be estimated by 
\begin{equation}\label{eqEstKappa}
	{\hat{\kappa} =  \frac{2C_\rho}{r}}\mbox{ .}
\end{equation}
\end{sloppypar}
Remark~\ref{rmkBinAppr} states that $C_\rho = \sum_{i=1}^{r}D_i \sim B(r,\frac{\kappa}{2})$ approximately follows a normal distribution $N\left(\frac{\kappa r}{2}, \frac{1}{2}r\kappa(1-\frac{\kappa}{2})\right)$.  Combining with Fact~\ref{factNormalFact}, we conclude that $\widehat{\kappa /2} =  \frac{C_\rho}{r} \sim N\left(\frac{\kappa}{2}, \frac{\kappa(1-\frac{\kappa}{2})}{2r}\right)$. Applying Fact~\ref{factNormalFact} again, we have ${\widehat{\kappa} \sim N\left(\kappa, \frac{2\kappa(1-\frac{\kappa}{2})}{r}\right)}$.

\subsection{Rate difference monitoring}\label{secRevProt}
When $\omega$ is unknown, we need to estimate it before starting the execution of the protocol. We have to assume that during this estimation, there is no attack. Under this assumption, $\kappa$ is reduced to $\omega$, the probability that a probing bit is disturbed. Correspondingly, $\hat{\kappa}$ is reduced to an estimator of $\omega$. Namely, $\kappa = \omega$ and $\hat{\kappa} = \hat{\omega}$. As we have shown ${\widehat{\kappa} \sim N\left(\kappa, \frac{2\kappa(1-\frac{\kappa}{2})}{r}\right)}$, we also have $\hat{\omega}\sim N\left(\omega, \frac{2\omega(1-\frac{\omega}{2})}{s}\right)$, where $s$ is the number of probing bits for estimating $\omega$. Let $C'_{\rho}$ denote the number of \glspl*{pm} under the assumption that the probing bits are not attacked. By replacing $C_{\rho}$ by $C'_{\rho}$ and $s$ by $r$ in equation~(\ref{eqEstKappa}), we get,
\begin{equation}\label{eqEstOmega}
	{\hat{\omega} =  \frac{2C'_{\rho}}{s}}\mbox{ .}
\end{equation}
In \gls*{renk}, the attack detection method is implemented by checking that $\kappa=\omega$. After getting the estimations of $\kappa$ and $\omega$, let $e$ denote their difference, which is an outcome of \gls*{rv} $E=\hat{\kappa}-\hat{\omega}$. Fact~\ref{factNormalAdd} states that $E$ still follows a normal distribution. In particular, $E\sim N\left(\kappa-\omega, \frac{2\kappa(1-\frac{\kappa}{2})}{r}+ \frac{2\omega(1-\frac{\omega}{2})}{s}\right)$. 
Under the null hypothesis $H_0$ that there is no attack, $\kappa =\omega$. Then $E\sim N\left(0,2\nu\left(1-\frac{1}{2}\nu\right)\left(\frac{1}{r}+ \frac{1}{s}\right)\right)$, where $\nu = \kappa=\omega$ and can be estimated by $\hat{\nu}=\frac{2(C'_{\rho} + C_{\rho})}{r+s}$. So, if $H_0$ is true, the distribution of \gls*{rv} $E$ is condensed near zero. Although the set of the possible outcomes of $E$ is $\mathbb{R}$, the test can rule out outcomes that are much greater than zero without introducing much error (note that we do not consider a negative difference because $\kappa$ is, theoretically, not less than $\omega$. In other words, the alternative hypothesis $H_1$ is $\kappa>\omega$).  Let $\alpha$ denote the probability that an outcome of $E$ is much greater than zero and ruled out by the test. We can test $H_0$ against $H_1$ by rejecting $H_0$ if we observe an outcome of $E$ greater than $e_\alpha$, where $e_\alpha\in \mathbb{R}$ such that $Pr[E>e_\alpha]=\alpha$. In other words, the protocol is considered insecure if $e$, the difference between the estimations of $\kappa$ and $\omega$, is greater than $e_\alpha$. 

The arduous calculation of $e_\alpha$ can be avoided if we scale $E$ to

\begin{equation}\label{eqCalcZ}
	Z=\frac{E-0}{\sqrt{2\nu\left(1-\frac{1}{2}\nu\right)\left(\frac{1}{r}+ \frac{1}{s}\right)}}=\frac{\hat{\kappa}- \hat{\omega}}{\sqrt{2\hat{\nu}\left(1-\frac{1}{2}\hat{\nu}\right)\left(\frac{1}{r}+ \frac{1}{s}\right)}}, 
\end{equation}
a standard normal distribution according to Fact~\ref{factStdNormal}. So correspondingly, the difference $e$ after scaling (denoted by $z$) is an outcome of $Z$. Then an equivalent test can be made by rejecting $H_0$ if $z>z_\alpha$ where $z_\alpha\in\mathbb{R}$ such that $Pr[Z>z_\alpha] = \alpha$. The table listing the value of $z_\alpha$ as a function of $\alpha$ can be found in Reference~\citep{hogg2014probability}. Therefore, we amend the original \gls*{nk} protocol as follows:
\begin{protocol}[\gls*{renk}]\label{protocolRev}
\
\begin{enumerate}[label=RC\arabic*]
		\item (Estimation of $\omega$) Alice and Bob execute \gls*{mrad} $s$ times.  Alice sends $s$ \glspl*{qb} to Bob. He reflects all of them. In other words, there are $s$ probing bits and zero data bits. In Step~\ref{NKAliceCheck}, Alice counts the number of \glspl*{pm}  (denoted by $C'_{\rho}$). Finally, she uses equation~(\ref{eqEstOmega}) to estimate $\omega$. Note that during the estimation process, we need to guarantee that Eve does not perpetrate attacks.
		\item\label{nkrStart} Alice and Bob start the execution of the protocol. They do Steps~\ref{NKAliceSend}-\ref{NKBobNotify}. 
		\item\label{nkrmeas} In \ref{NKAliceCheck}, instead of terminating the protocol when $p_k = 0$ and function $AE(e_{1k},e_{2k},i_k) = 1$, Alice increments a counter $C_\rho$ (initial value is zero) and continues to check the remaining bits of $\undertilde{P}$. After finishing checking, she uses equation~(\ref{eqEstKappa}) to estimate $\kappa$. We test the null hypothesis  $H_0: \kappa = \omega$ against the alternative hypothesis $H_1: \kappa > \omega$, equation~(\ref{eqCalcZ}). If $H_0$ is rejected, Alice considers the protocol is insecure and terminates it; otherwise, Alice and Bob execute Steps~\ref{NKAliceFetch} and \ref{NKCorrect} to complete the data transmission.
	\end{enumerate}
\end{protocol}

When $\omega$ is given, we can simply compare it with the estimation of $\kappa$. Similarly, we need to test $H_0: \kappa = \omega$ against $H_1: \kappa > \omega$. Since  $\omega$ is not estimated but a given constant, we can say $\hat{\omega} \sim N(\omega, 0)$. We estimate the real attack rate $\kappa$ by equation~(\ref{eqEstKappa}). Applying Fact~\ref{factNormalAdd}, we have that $E=\hat{\kappa} - \omega = \hat{\kappa} - \hat{\omega} \sim N\left(\kappa-\omega, \frac{2\kappa(1-\frac{\kappa}{2})}{r}\right)$. Under the assumption that $H_0$ is true, $E\sim N\left(0, \frac{2\kappa(1-\frac{\kappa}{2})}{r}\right)$. Applying Fact~\ref{factStdNormal}, we scale $E$ to $Z' = \frac{(\hat{\kappa} - \omega)-0}{\sqrt{\frac{2\kappa(1-\frac{\kappa}{2})}{r}}}= \frac{\hat{\kappa} - \omega}{\sqrt{\frac{2\kappa(1-\frac{\kappa}{2})}{r}}}\sim N(0,1)$. Let $z'$ denote the scaled difference of the estimated $\kappa$ and the pre-known $\omega$, which is an outcome of $Z'$. We reject $H_0$ if $z' > z_\alpha$, where the definition of  $z_\alpha$ is unchanged. 

Since $\omega$ is given, its estimation is unnecessary. To complete the data transmission, Alice and Bob only need to implement Steps~\ref{nkrStart} and \ref{nkrmeas}, where $Z$ is replaced by $Z'$.

\subsection{Security analysis of \gls*{renk}}
As we have mentioned at the beginning of this section, the original \gls*{mrad} fails if the \glspl*{qb} transmitted are disturbed and the entanglement is destroyed. The false positives mislead the protocol about the transmission security and cause wrong termination. To fix the problem, in Sections~\ref{secRevProt}, we propose a new detection method for \glspl*{attackmr} based on a statistical test. The method detects the  discrepancy between $\omega$ and $\kappa$, which does not exist if there is no attack. If any significant discrepancy is identified, the protocol is considered insecure and terminated.

The test rules out the possible outcomes of $E$ that are largely greater than zero, and thus, introduces detection errors. In more details, suppose that Eve does not perpetrate attacks, which implies $\kappa = \omega$ and the null hypothesis $H_0$ is true. Due to the fluctuation of the estimator $D$, the difference between $\hat{\kappa}$ and $\hat{\omega}$, there is a proabability $\alpha$ that the sampling of $D$ is greater than the threshold $d_\alpha$ and gets the $H_0$ rejected, which is a false positive. Correspondingly, if Eve perpetrates attacks and causes $\kappa>\omega$, it is also possible that $H_0$ is not rejected since their difference is still less than $d_\alpha$, which is a false negative. We formally define these two types of errors as follows, 

\begin{definition}[Type A Error - False Negative]
	Eve perpetrates an attack, but the protocol is wrongly considered secure.
\end{definition}

\begin{definition}[Type B Error - False Positive]
	Eve does not perpetrate an attack, but the protocol is wrongly considered insecure.
\end{definition}
The Type A Error has more adverse consequences than the Type B Error because Eve can eavesdrop the message without the awareness of  Alice and Bob. We show that the probability of undetected eavesdropping is very low, even when a small number of probing bits is used. The Type B error does not undermine the security of the protocol. Instead, it lowers the transmission efficiency. While Eve does not perpetrate an attack, the Type B error causes a wrong belief of its presence and a termination of the protocol. The protocol needs to restart and resend all \glspl*{qb}. The transmission efficiency is affected. 

The choice of a specific value for $\alpha$, the occurrence probability of the Type B Error, affects the one of the Type A Error. In particular, an increase of $\alpha$ pushes the value of $d_\alpha$ to zero. Although Eve only attacks a few portion of the probing bits, the difference between $\kappa$ and $\omega$ she introduces may still exceed the lowered $d_\alpha$ and get $H_0$ rejected; therefore, the test becomes stricter and the occurrence probability of  Type A Error decreases. Similarly, we can show that a decrease of $\alpha$ leads to an increase of Type A Error occurrence probability. Since the Type A and B Error occurrence probabilities have a negative relationship, if we increase $\alpha$ to enhance the security level, we get more Type B Errors and lower transmission efficiency. Conversely, to decrease the overhead, security is undermined. 

\begin{figure*}[h]
	\centering
    \includegraphics[width=0.8\columnwidth]{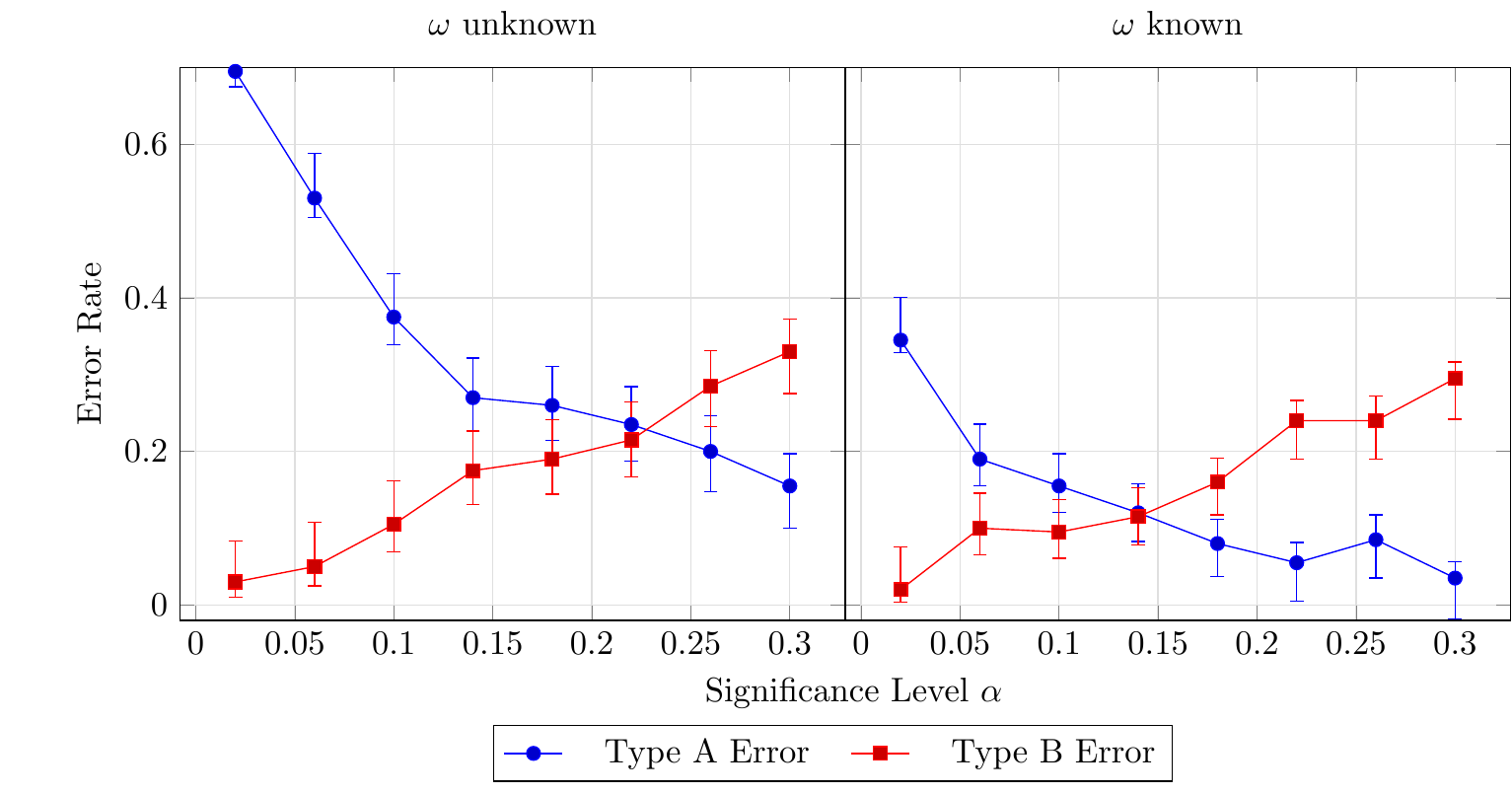}
	\caption{Error rates and their $95\%$ confidence intervals as a function of significance level when the probability $\omega$ that a \gls*{qb} is disturbed is unknown. ($\omega = 0.3$) is unknown~(left) and known~(right). (Simulation configuration: $r = 600$, $s = 600$ (if $\omega$ is unknown), $\omega = 0.3$, $p = 0.1$)}\label{figAlpha}
\end{figure*}

\begin{figure*}[h!]
	\centering
	\includegraphics[width=0.75\columnwidth]{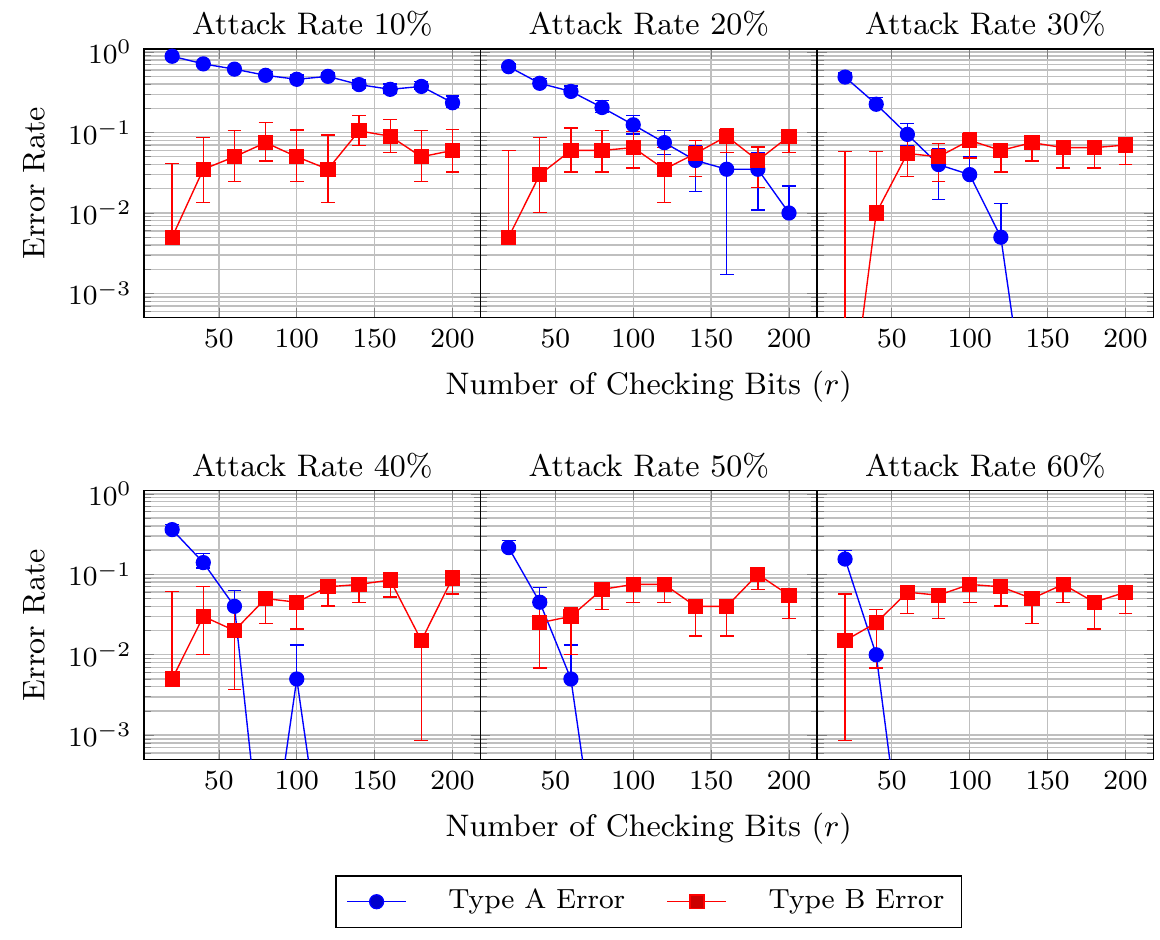}
    \caption{Error rates and their $95\%$ confidence intervals with respect to the number of probing bits when Eve has probability $p = 0.1$, $0.2$, ..., $0.6$ to attack a \gls*{qb}. (The probability $\omega$ that a \gls*{qb} is disturbed is unknown. Simulation configuration: $\omega = 0.05$, $\alpha = 0.05$)}\label{figDErrCheckU}
\end{figure*}

\begin{figure*}[h!]
	\centering
    \includegraphics[width=0.75\columnwidth]{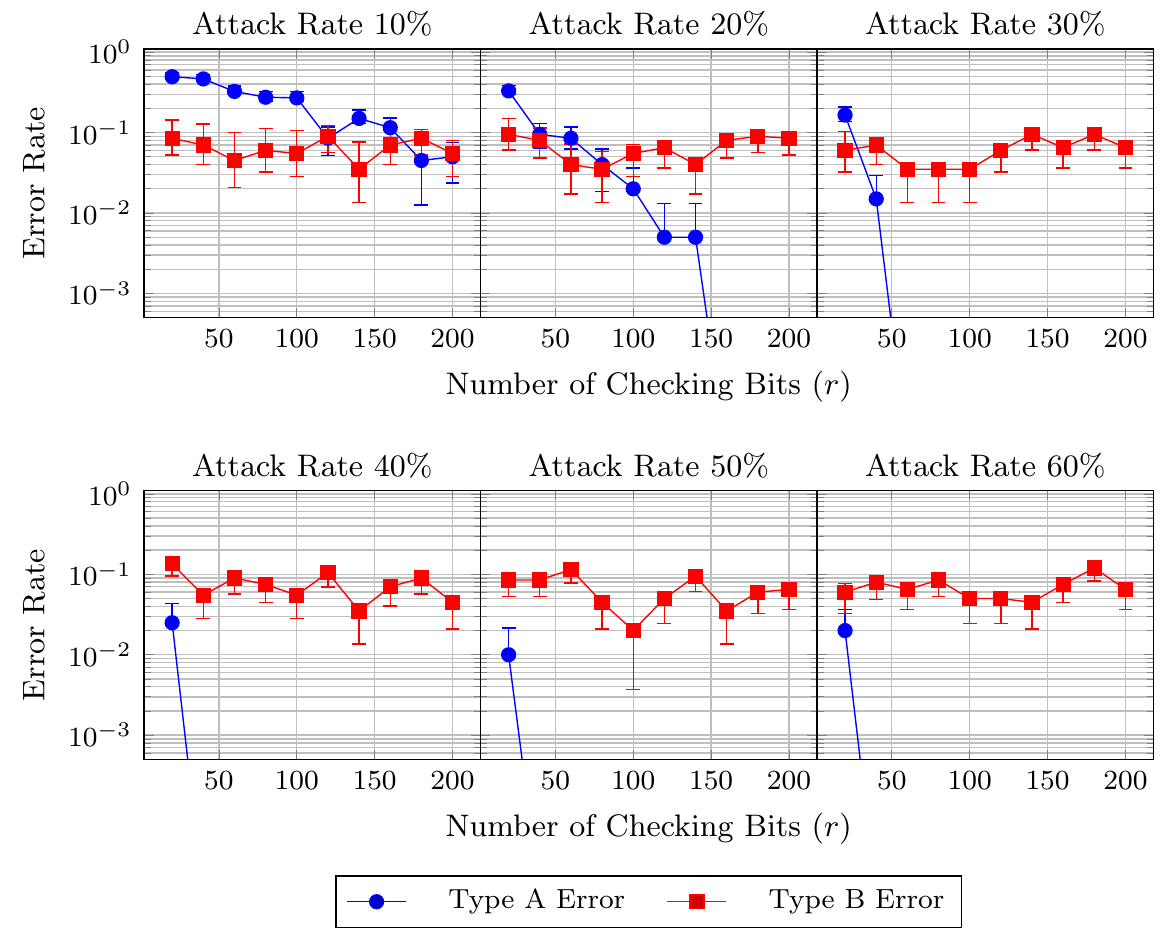}
    \caption{Error rates and their $95\%$ confidence intervals with respect to the number of probing bits when Eve has probability $p = 0.1$, $0.2$, ..., $0.6$ to attack a \gls*{qb}. (The probability $\omega$ that a \gls*{qb} is disturbed is known. Simulation configuration: $\omega = 0.05$, $\alpha = 0.05$)}\label{figDErrCheckK}
\end{figure*}

With the results of simulations, Figure~\ref{figAlpha} plots the rates of the two types of errors as a function of $\alpha$. For estimating the Type A Error occurrence probability, we set the probability $p$ for Eve to attack a \gls*{qb} to $10\%$ for both cases, and the number of probing bits to estimate $\kappa$ and $\omega$ (if unknown) to $600$. Note that the configuration here is intended to make the Type A Error occurrence probability more sensitive to the choice of alpha, which is not typical in practical problems. We will talk about how the error occurrence probabilities behave with more common configurations in the sequel.  Whether the value $\omega$ is known or not, the trends for both types are consistent with our analysis. When $\omega$ is given, the probability of the Type A Error is lower because the estimation of $\omega$ introduces more variance, which further amplifies the fluctuation of the estimation of the difference $D = \kappa -\omega$. Regarding the Type B Error, we can observe that the rate roughly equals $\alpha$ which makes sense since it is an estimation of it. 

Besides the occurrence probability $\alpha$ of the Type B Error, the numbers of probing bits required to estimate $\omega$ and $\kappa$ are also related to the transmission efficiency. A larger number of probing bits contributes to a better estimation, but also has higher overhead. With the results of simulations, Figures~\ref{figDErrCheckU} and \ref{figDErrCheckK} plot the rates of the Type A and B Errors as a function of the number of probing bits and the attack rate. The probability (${ \omega = 0.05}$) that a \gls*{qb} is disturbed is unknown in Figure~\ref{figDErrCheckU} but pre-known in Figure~\ref{figDErrCheckK}. $\alpha$ is set to $0.05$. In Figure~\ref{figDErrCheckU}, the numbers of probing bits coincide for the estimation of $\omega$ and $\kappa$.

The two figures show that when Eve is more likely to attack a \gls*{qb}, the detection success rate increases. If Eve only attacks a small proportion of \glspl*{qb}, her attacks do not significantly increase $\kappa$ and thus are concealed by $\omega$. However, in order to successfully eavesdrop messages, Eve should perpetrates attacks at a rate higher than $50\%$. When the probability of attacks is $60\%$, $60$ probing bits are sufficient to avoid the Type A Error (when $\omega$ is unknown). If $\omega$ is given, then $40$ probing bits can achieve the same security level.

Note that the Type B Error rate should be constant. In particular, its mean is theoretically equal to $5\%$ as it is an estimation of $\alpha$. However, while the estimated rate roughly stays around $5\%$ in Figure~\ref{figDErrCheckK}, a relatively considerable increase is observed in Figure~\ref{figDErrCheckU}. The increase is due to a low number of probing bits. According to Remark~\ref{rmkBinAppr}, a good normal approximation requires a large sample size and to estimate both $\omega$ and $\kappa$, a even larger one is needed. Although, the approximation is not quite accurate when the probing bit number is small, a low level of Type A Error rate shows that it is good enough to secure the protocol.

Since the \gls*{renk} and \gls*{nk} protocols are the same except for the part that detects \glspl*{attackmr}, Theorem~\ref{thmNkSecure} is also applicable to \gls*{renk}. In particular, we have Theorem~\ref{thmNoAErr}. 
\begin{theorem}\label{thmNoAErr}
	With a sufficient number of probing bits, the Type A Error can be avoided. So the  \gls*{renk} protocol is secure.
\end{theorem}

\section{Quantum Resource Requirements and Transmission Efficiency}\label{secEvaluation}
In this section, we analyze the requirements of quantum resources, \gls*{qb} efficiency and quantum circuit complexity of the \gls*{nk} protocol. Among all \gls*{sqkd} and \gls*{sqdc} protocols, we show that the \gls*{nk} protocol has the highest \gls*{qb} efficiency (almost $100\%$) with the simplest quantum circuits (without \glspl*{qb} permutation and measurement basis switch). Comparing to the protocols utilizing the quantum entanglements, we show that the \gls*{nk} protocol reaches the theoretical minimum of the \gls*{qbr} size and the \gls*{ept} among the \gls*{sqkd} and \gls*{sqdc} protocols.

\subsection{Quantum resources requirements}\label{secQuanRec}
We briefly discuss the quantum resources requirements at the end of Section~\ref{nk}. In this section, we elaborate them. Since the revised version introduced in Section~\ref{secRev} does not change the hardware requirements, we discuss them together.

Alice and Bob loop over $i_k$ in $I$ to accomplish all the operations requiring quantum resources.  Hence, the lowest quantum resources requirements of the protocol implementation agrees with the one to execute a single quantum procedure group $QPG_k$ (plotted in Figure~\ref{figqpk}).

\begin{figure}[!h]
	\centering
	\includegraphics[width=0.8\columnwidth]{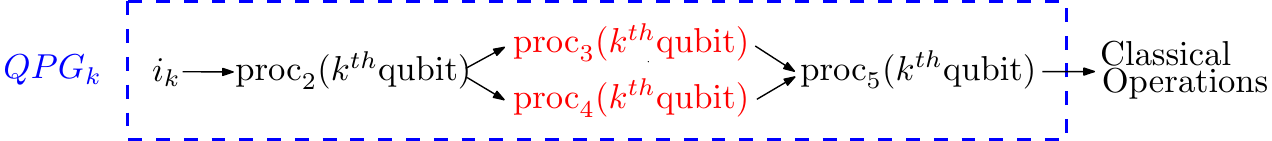}
	\caption{The \gls*{nk} protocol flow diagram of the procedures requiring quantum resources. \textit{The procedures marked in red belong to Bob. For each $QPG_k$, Bob chooses either $proc_3(k^{th} \mbox{ \gls*{qb}})$ to measure or $proc_4(k^{th} \mbox{ \gls*{qb}})$ to reflect}.}\label{figqpk}
\end{figure}
In order to generate EPR pairs in Procedure~\ref{OAS}, Alice is required to have an EPR pair generation circuit. Besides, a one-bit \gls*{qbr} is needed to retain the first \gls*{qb} of the pair. Alice needs a Bell measurement circuit to perform Procedure~\ref{OAM}. The entanglement of the EPR pair must be preserved until Alice applies a Bell measurement on it (this case happens when the second \gls*{qb} $\gamma_B$ is reflected by Bob. Otherwise, Bob measures it, and the preservation time is shorter). Let $C$ denote the time that Alice generates, sends and receives the \glspl*{qb}, and $T$ the one-way time for the \glspl*{qb} to move between Alice and Bob. Then the \gls*{ept} is $C+2T$ if we do not count the \glspl*{qb} reflection time of Bob. Compared to Alice, the quantum capability of Bob is fundamental. In particular, he should be able to either measure a \gls*{qb} in the Z-basis, followed by sending a pre-generated $\ket{0}$, or reflect it. Overall, Alice and Bob need the following minimum quantum capabilities for the protocol implementation. 

\noindent\textbf{Alice:} a one-bit \gls*{qbr}, the circuits for Bell measurement and EPR pair generation.\\
\noindent\textbf{Bob:} a device that either reflects a \gls*{qb} or uses the Z-basis to measure it  followed by sending a $\ket{0}$.

Among \gls*{sqdc} and \gls*{sqkd} protocols that utilize entanglements, Alice must create at least a pair of entangled \glspl*{qb} and send at least one of the \glspl*{qb} to Bob. Therefore, for containing a \gls*{qb}, a one-bit \gls*{qbr} is necessary. For checking the potential attacks, Alice must do some quantum operations on the \gls*{qb} pair consisting of the \gls*{qb} she retained and the one reflected by Bob. So, the \gls*{ept} is at least $C+2T$. As the \gls*{nk} reaches the theoretical lower bound, we conclude that,
\begin{theorem}
	Among \gls*{sqdc} and \gls*{sqkd} protocols that utilize entanglements, the \gls*{nk} protocol only requires theoretically minimal \gls*{qbr} size and \gls*{ept}. 
\end{theorem}

\subsection{Transmission efficiency}
Suppose the string $\undertilde{C}$ $(=\undertilde{U})$ shared by Alice and Bob is updated for each dialogue. In other words, prior to applying Step~\ref{NKAliceFetch} and Step~\ref{NKCorrect} to transmit data, Alice and Bob always execute Steps~\ref{NKAliceSend} to \ref{NKAliceCheck} to share a random binary string. 

Consider the original \gls*{nk} protocol. Suppose that an $s$-bit message is sent to Bob from Alice, and for each \gls*{qb} Eve has probability $p$ to perpetrate an \gls*{attackmr}. Then, for a high eavesdropping efficiency, Eve has to choose a $p$ close to one. Assume $p = 0.6$ and the number of probing bits is $15$. Theorem~\ref{thmDetectMRAp} shows that the success rate of detection is higher than $0.995$. Moreover, adding a few more probing bits can enhance the security level significantly. In practice, the message length $s$ should considerably exceed $15$. Then the probing bits can only introduce a negligible overhead. So for sending a message of length $s$, Alice sends roughly $s$ qubits to Bob. Thus, the \gls*{qb} efficiency approaches $100\%$.

 If we consider the possibility that a \gls*{qb} is disturbed during the \gls*{qb} transmission, we need to apply the rate estimation version of the protocol. If $\omega$ is unknown, we need around $60$ probing bits to reach $99\%$ detection success rate (assuming that $\alpha = 0.01$, $\omega = 0.05$), which is acceptable considering a much larger total number of \glspl*{qb} transmitted. If the rate is given, then the probing bit number can decrease to $30$ (extra $20$ probing bits can improve the success rate to almost $100\%$). Then the overhead from the probing bits is negligible. The major part of the overhead is from  $\alpha$, the probability to get Type B Error, which causes a full restart of the protocol. In average, $\alpha\cdot 100\%$ \glspl*{qb} transmitted are discarded due to the wrong conclusion that the protocol is insecure. If we choose $\alpha = 0.01$ (which is big enough to secure the protocol), the overhead is only $1\%$.

Notice that Alice is only required to perform Bell measurements. So she only needs a fixed circuit without measurement basis switch capability. Additionally, the operations related to one \gls*{qb} is irrelevant to those concerned with the others (since the message security does not depend on the bit permutation by Alice and Bob). Therefore, if the transmission or measurement of a single \gls*{qb} fails, Alice and Bob only need to re-implement the operations associated with that \gls*{qb}. This enhances the success rate and efficiency of the data transmission potentially.

In Table~\ref{tableCompare}, we make a detailed comparison with other typical \gls*{sqkd} and \gls*{sqdc} protocols. Note that the \gls*{qb} efficiency ($\eta$) is calculated by $$\eta=\frac{\mbox{Length of the message}}{\mbox{Number of \glspl*{qb} sent by Alice}}. $$ For the protocols in References \citep{PhysRevA.79.032341,0256-307X-28-10-100301,PhysRevA.79.052312,li2016}, $\eta$ depends on some parameters other than the length of the message. For these protocols, we give an upper bound for $\eta$. Regarding \gls*{renk}, $\eta$ is calculated by choosing $\alpha = 0.01$. Besides, in the protocol proposed by Li et al. \citep{li2016}, the measurement basis switch is not required of Alice or Bob but delegated to a third full quantum capability computer Charlie.

\section{Conclusion}\label{secConclusion}
In this paper, we proposed a new \gls*{sqdc} protocol (named \acrlong*{nk}). Compared to other \gls*{sqdc} and \gls*{sqkd} protocols, our new protocol has much higher \gls*{qb} efficiency (almost $100\%$) and simpler quantum circuits (not requiring switching measurement basis or permuting  \glspl*{qb}). While other \gls*{sqkd} and \gls*{sqdc} protocols encrypting messages through entanglements require at least linear \gls*{ept} and linear size \gls*{qbr}, in our protocol, only Alice is required to have a fixed size (as low as one) \gls*{qbr} and preserve an EPR pair entanglement for time $C+2T$, where $C$ is the time that Alice prepares, receives and measures the \glspl*{qb}, and $T$ is the one for the \glspl*{qb} to move between Alice and Bob. 

Among the protocols using quantum entanglements to encrypt messages, we show that both \gls*{qbr} size and \gls*{ept} achieve the theoretical minimums. A pre-shared key is not required by our new protocol. Instead, Alice and Bob use the \glspl*{qb} entanglement to share a random string and further use it as a key to secure the data transmission. We used the probing bits to implement \gls*{mrad} so that the protocol is resistant to \glspl*{attackmr}. 

In our original protocol, Theorem~\ref{thmDetectMRAp} shows that $15$ probing bits can lead to a $0.995$ success rate of attack detection (given that the adversary Eve has the probability $0.6$ of perpetrating an \gls*{attackmr} on a single \gls*{qb}). A few more bits can boost the security level of the protocol significantly (for example, $0.9992$ detection success rate can be achieved by using $20$ probing bits). If the message size is sufficiently long, then the \gls*{qb} efficiency can reach almost $100\%$.

The rate estimation version, the protocol \gls*{renk}, can function properly and correctly detect attacks perpetrated by Eve while the \glspl*{qb} may be disturbed during the transmission. We designed a test to monitor the difference of $\kappa$, the probability that a \gls*{qb} is disturbed or attacked, and $\omega$ (estimated or pre-known), the probability that a \gls*{qb} is disturbed. If the difference is significantly large, the protocol terminates. The simulation results show that $60$ probing bits can push detection success rate to almost $100\%$ (assuming that $\alpha = 0.05$ and $\omega = 0.05$) if $\omega$ is unknown.  The number of probing bits can decrease to $40$ and achieve the same success rate if $\omega$ is pre-known. Assuming that we can always detect \glspl*{attackmr}, our protocol is secure against network attacks (Theorems~\ref{thmNkSecure} and \ref{thmNoAErr}).

\bibliographystyle{spmpsci}      
\bibliographystyle{spphys}    
\bibliography{Bibliography}

\begin{thebibliography}{10}
\expandafter\ifx\csname url\endcsname\relax
  \def\url#1{\texttt{#1}}\fi
\expandafter\ifx\csname urlprefix\endcsname\relax\def\urlprefix{URL }\fi
\expandafter\ifx\csname doiprefix\endcsname\relax\def\doiprefix{DOI }\fi
\providecommand{\bibinfo}[2]{#2}
\providecommand{\eprint}[2][]{\url{#2}}

\bibitem{PhysRevLett.85.441}
\bibinfo{author}{Shor, P.~W.} \& \bibinfo{author}{Preskill, J.}
\newblock \bibinfo{journal}{\bibinfo{title}{Simple proof of security of the
  bb84 quantum key distribution protocol}}.
\newblock {\emph{\JournalTitle{Phys. Rev. Lett.}}}
  \textbf{\bibinfo{volume}{85}}, \bibinfo{pages}{441--444}
  (\bibinfo{year}{2000}).

\bibitem{PhysRevLett.67.661}
\bibinfo{author}{Ekert, A.~K.}
\newblock \bibinfo{journal}{\bibinfo{title}{Quantum cryptography based on
  bell's theorem}}.
\newblock {\emph{\JournalTitle{Phys. Rev. Lett.}}}
  \textbf{\bibinfo{volume}{67}}, \bibinfo{pages}{661--663}
  (\bibinfo{year}{1991}).

\bibitem{PhysRevA.65.032302}
\bibinfo{author}{Long, G.~L.} \& \bibinfo{author}{Liu, X.~S.}
\newblock \bibinfo{journal}{\bibinfo{title}{Theoretically efficient
  high-capacity quantum-key-distribution scheme}}.
\newblock {\emph{\JournalTitle{Phys. Rev. A}}} \textbf{\bibinfo{volume}{65}},
  \bibinfo{pages}{032302} (\bibinfo{year}{2002}).

\bibitem{Lo2050}
\bibinfo{author}{Lo, H.-K.} \& \bibinfo{author}{Chau, H.~F.}
\newblock \bibinfo{journal}{\bibinfo{title}{Unconditional security of quantum
  key distribution over arbitrarily long distances}}.
\newblock {\emph{\JournalTitle{Science}}} \textbf{\bibinfo{volume}{283}},
  \bibinfo{pages}{2050--2056} (\bibinfo{year}{1999}).

\bibitem{doi:10.1137/S0097539795293172}
\bibinfo{author}{Shor, P.~W.}
\newblock \bibinfo{journal}{\bibinfo{title}{Polynomial-time algorithms for
  prime factorization and discrete logarithms on a quantum computer}}.
\newblock {\emph{\JournalTitle{SIAM Journal on Computing}}}
  \textbf{\bibinfo{volume}{26}}, \bibinfo{pages}{1484--1509}
  (\bibinfo{year}{1997}).

\bibitem{PhysRevLett.99.140501}
\bibinfo{author}{Boyer, M.}, \bibinfo{author}{Kenigsberg, D.} \&
  \bibinfo{author}{Mor, T.}
\newblock \bibinfo{journal}{\bibinfo{title}{Quantum key distribution with
  classical bob}}.
\newblock {\emph{\JournalTitle{Phys. Rev. Lett.}}}
  \textbf{\bibinfo{volume}{99}}, \bibinfo{pages}{140501}
  (\bibinfo{year}{2007}).

\bibitem{PhysRevA.79.032341}
\bibinfo{author}{Boyer, M.}, \bibinfo{author}{Gelles, R.},
  \bibinfo{author}{Kenigsberg, D.} \& \bibinfo{author}{Mor, T.}
\newblock \bibinfo{journal}{\bibinfo{title}{Semiquantum key distribution}}.
\newblock {\emph{\JournalTitle{Phys. Rev. A}}} \textbf{\bibinfo{volume}{79}},
  \bibinfo{pages}{032341} (\bibinfo{year}{2009}).

\bibitem{0256-307X-28-10-100301}
\bibinfo{author}{Jian, W.}, \bibinfo{author}{Sheng, Z.}, \bibinfo{author}{Quan,
  Z.} \& \bibinfo{author}{Chao-Jing, T.}
\newblock \bibinfo{journal}{\bibinfo{title}{Semiquantum key distribution using
  entangled states}}.
\newblock {\emph{\JournalTitle{Chinese Physics Letters}}}
  \textbf{\bibinfo{volume}{28}}, \bibinfo{pages}{100301}
  (\bibinfo{year}{2011}).

\bibitem{li2016}
\bibinfo{author}{Li, Q.}, \bibinfo{author}{Chan, W.~H.} \&
  \bibinfo{author}{Zhang, S.}
\newblock \bibinfo{journal}{\bibinfo{title}{Semiquantum key distribution with
  secure delegated quantum computation}}.
\newblock {\emph{\JournalTitle{Scientific Reports}}}
  \textbf{\bibinfo{volume}{6}} (\bibinfo{year}{2016}).

\bibitem{Luo2016}
\bibinfo{author}{Luo, Y.-P.} \& \bibinfo{author}{Hwang, T.}
\newblock \bibinfo{journal}{\bibinfo{title}{Authenticated semi-quantum direct
  communication protocols using bell states}}.
\newblock {\emph{\JournalTitle{Quantum Inf. Process.}}}
  \textbf{\bibinfo{volume}{15}}, \bibinfo{pages}{947--958}
  (\bibinfo{year}{2016}).

\bibitem{7848870}
\bibinfo{author}{Almousa, S.} \& \bibinfo{author}{Barbeau, M.}
\newblock \bibinfo{journal}{\bibinfo{title}{Delay and reflection attacks in
  authenticated semi-quantum direct communications}}.
\newblock {\emph{\JournalTitle{2016 IEEE Globecom Workshops (GC Wkshps)}}}
  \bibinfo{pages}{1--7} (\bibinfo{year}{2016}).

\bibitem{Shukla2017}
\bibinfo{author}{Shukla, C.}, \bibinfo{author}{Thapliyal, K.} \&
  \bibinfo{author}{Pathak, A.}
\newblock \bibinfo{journal}{\bibinfo{title}{Semi-quantum communication:
  protocols for key agreement, controlled secure direct communication and
  dialogue}}.
\newblock {\emph{\JournalTitle{Quantum Information Processing}}}
  \textbf{\bibinfo{volume}{16}}, \bibinfo{pages}{295} (\bibinfo{year}{2017}).

\bibitem{PhysRevLett.118.220501}
\bibinfo{author}{Zhang, W.} \emph{et~al.}
\newblock \bibinfo{journal}{\bibinfo{title}{Quantum secure direct communication
  with quantum memory}}.
\newblock {\emph{\JournalTitle{Phys. Rev. Lett.}}}
  \textbf{\bibinfo{volume}{118}}, \bibinfo{pages}{220501}
  (\bibinfo{year}{2017}).

\bibitem{Wu2017}
\bibinfo{author}{Wu, F.} \emph{et~al.}
\newblock \bibinfo{journal}{\bibinfo{title}{High-capacity quantum secure direct
  communication with two-photon six-qubit hyperentangled states}}.
\newblock {\emph{\JournalTitle{Science China Physics, Mechanics {\&}
  Astronomy}}} \textbf{\bibinfo{volume}{60}}, \bibinfo{pages}{120313}
  (\bibinfo{year}{2017}).

\bibitem{Gu2018}
\bibinfo{author}{Gu, J.}, \bibinfo{author}{Lin, P.-h.} \&
  \bibinfo{author}{Hwang, T.}
\newblock \bibinfo{journal}{\bibinfo{title}{Double c-not attack and
  counterattack on `three-step semi-quantum secure direct communication
  protocol'}}.
\newblock {\emph{\JournalTitle{Quantum Information Processing}}}
  \textbf{\bibinfo{volume}{17}}, \bibinfo{pages}{182} (\bibinfo{year}{2018}).

\bibitem{2015Zhong}
\bibinfo{author}{Zhong, M.} \emph{et~al.}
\newblock \bibinfo{journal}{\bibinfo{title}{Optically addressable nuclear spins
  in a solid with a six-hour coherence time}}.
\newblock {\emph{\JournalTitle{Nature}}} \textbf{\bibinfo{volume}{517}},
  \bibinfo{pages}{177--180} (\bibinfo{year}{2015}).

\bibitem{Inagaki:13}
\bibinfo{author}{Inagaki, T.}, \bibinfo{author}{Matsuda, N.},
  \bibinfo{author}{Tadanaga, O.}, \bibinfo{author}{Asobe, M.} \&
  \bibinfo{author}{Takesue, H.}
\newblock \bibinfo{journal}{\bibinfo{title}{Entanglement distribution over 300
  km of fiber}}.
\newblock {\emph{\JournalTitle{Opt. Express}}} \bibinfo{pages}{23241--23249}.

\bibitem{2010Neumann}
\bibinfo{author}{Neumann, P.} \emph{et~al.}
\newblock \bibinfo{journal}{\bibinfo{title}{Quantum register based on coupled
  electron spins in a room-temperature solid}}.
\newblock {\emph{\JournalTitle{Nat Phys}}} \textbf{\bibinfo{volume}{6}},
  \bibinfo{pages}{249--253} (\bibinfo{year}{2010}).

\bibitem{2016Dai}
\bibinfo{author}{Dai, H.-N.} \emph{et~al.}
\newblock \bibinfo{journal}{\bibinfo{title}{Generation and detection of atomic
  spin entanglement in optical lattices}}.
\newblock {\emph{\JournalTitle{Nat Phys}}} \textbf{\bibinfo{volume}{12}},
  \bibinfo{pages}{783--787} (\bibinfo{year}{2016}).

\bibitem{8269077}
\bibinfo{author}{Lu, H.}, \bibinfo{author}{Barbeau, M.} \&
  \bibinfo{author}{Nayak, A.}
\newblock \bibinfo{title}{Economic no-key semi-quantum direct communication
  protocol}.
\newblock In \emph{\bibinfo{booktitle}{2017 IEEE Globecom Workshops (GC
  Wkshps)}}, \bibinfo{pages}{1--7} (\bibinfo{year}{2017}).

\bibitem{Malladi02onpreventing}
\bibinfo{author}{Malladi, S.}, \bibinfo{author}{Alves-Foss, J.} \&
  \bibinfo{author}{Heckendorn, R.~B.}
\newblock \bibinfo{journal}{\bibinfo{title}{On preventing replay attacks on
  security protocols}}.
\newblock {\emph{\JournalTitle{In Proc. Int. Conf. on Security and
  Management}}} \bibinfo{pages}{77--83} (\bibinfo{year}{2002}).

\bibitem{PhysRevLett.92.057901}
\bibinfo{author}{Scarani, V.}, \bibinfo{author}{Ac\'{\i}n, A.},
  \bibinfo{author}{Ribordy, G.} \& \bibinfo{author}{Gisin, N.}
\newblock \bibinfo{journal}{\bibinfo{title}{Quantum cryptography protocols
  robust against photon number splitting attacks for weak laser pulse
  implementations}}.
\newblock {\emph{\JournalTitle{Phys. Rev. Lett.}}}
  \textbf{\bibinfo{volume}{92}}, \bibinfo{pages}{057901}
  (\bibinfo{year}{2004}).

\bibitem{PhysRevLett.94.230504}
\bibinfo{author}{Lo, H.-K.}, \bibinfo{author}{Ma, X.} \& \bibinfo{author}{Chen,
  K.}
\newblock \bibinfo{journal}{\bibinfo{title}{Decoy state quantum key
  distribution}}.
\newblock {\emph{\JournalTitle{Phys. Rev. Lett.}}}
  \textbf{\bibinfo{volume}{94}}, \bibinfo{pages}{230504}
  (\bibinfo{year}{2005}).

\bibitem{Kalashnikov:11}
\bibinfo{author}{Kalashnikov, D.~A.}, \bibinfo{author}{Tan, S.~H.},
  \bibinfo{author}{Chekhova, M.~V.} \& \bibinfo{author}{Krivitsky, L.~A.}
\newblock \bibinfo{journal}{\bibinfo{title}{Accessing photon bunching with a
  photon number resolving multi-pixel detector}}.
\newblock {\emph{\JournalTitle{Opt. Express}}} \textbf{\bibinfo{volume}{19}},
  \bibinfo{pages}{9352--9363} (\bibinfo{year}{2011}).

\bibitem{Zhou:14}
\bibinfo{author}{Zhou, Z.} \emph{et~al.}
\newblock \bibinfo{journal}{\bibinfo{title}{Superconducting series nanowire
  detector counting up to twelve photons}}.
\newblock {\emph{\JournalTitle{Opt. Express}}} \textbf{\bibinfo{volume}{22}},
  \bibinfo{pages}{3475--3489} (\bibinfo{year}{2014}).

\bibitem{hogg2014probability}
\bibinfo{author}{Hogg, R.}, \bibinfo{author}{Tanis, E.} \&
  \bibinfo{author}{Zimmerman, D.}
\newblock \emph{\bibinfo{title}{Probability and Statistical Inference}},
  \bibinfo{pages}{192} (\bibinfo{publisher}{Pearson Education},
  \bibinfo{year}{2014}).

\bibitem{DasGupta2010}
\bibinfo{author}{DasGupta, A.}
\newblock \emph{\bibinfo{title}{Normal Approximations and the Central Limit
  Theorem}}, \bibinfo{pages}{213--242} (\bibinfo{publisher}{Springer New York},
  \bibinfo{address}{New York, NY}, \bibinfo{year}{2010}).

\bibitem{PhysRevA.79.052312}
\bibinfo{author}{Zou, X.}, \bibinfo{author}{Qiu, D.}, \bibinfo{author}{Li, L.},
  \bibinfo{author}{Wu, L.} \& \bibinfo{author}{Li, L.}
\newblock \bibinfo{journal}{\bibinfo{title}{Semiquantum-key distribution using
  less than four quantum states}}.
\newblock {\emph{\JournalTitle{Phys. Rev. A}}} \textbf{\bibinfo{volume}{79}},
  \bibinfo{pages}{052312} (\bibinfo{year}{2009}).

\end{thebibliography}

\section*{Author Contributions}
H.L devised the protocol. H.L and M.B wrote the main manuscript, and A.N gave solid suggestions on the manuscript. All authors reviewed the manuscript.
\section*{Additional Information}
\textbf{Competing interests:} The authors declare no competing interests.

% Table 1
\begin{table}[!h]
\centering
\caption{Bell measurements on Bell states.}\label{tableBMBS}
\begin{tabular}{||c|c||c|c||}
			\hline
			  $\gamma_A\gamma_B$  &   Output $(e_1e_2)$   &   $\gamma_A\gamma_B$   &   Output $(e_1e_2)$   \\\hline\hline

            $\ket{\Phi^+}$ & $00$ & $\ket{\Phi^-}$ & $10$ \\\hline
            $\ket{\Psi^+}$ & $01$ & $\ket{\Psi^-}$ & $11$ \\\hline
        \end{tabular}
\end{table}

% Table 2
\begin{table}[!h]
\centering
\caption{Bell measurement results of $\ket{00}$, $\ket{01}$, $\ket{10}$ and $\ket{11}$; there is $0.5$ possibility for each output.}\label{tableBMNBS}
\begin{tabular}{||c|c||c|c||}
			\hline
			  $\gamma_A\gamma_B$  &   Output $(e_1e_2)$   &   $\gamma_A\gamma_B$   &   Output $(e_1e_2)$   \\\hline\hline
            \multirow{2}{*}{$\ket{00}$} & $00$ & \multirow{2}{*}{$\ket{01}$} & $01$ \\
            \cline{2-2} \cline{4-4}& $10$ &  & $11$ \\\hline
            \hline
            \multirow{2}{*}{$\ket{10}$} & $01$ & \multirow{2}{*}{$\ket{11}$} & $00$ \\
            \cline{2-2} \cline{4-4}& $11$ &  & $10$ \\\hline
        \end{tabular}
\end{table}

% Table 3
\begin{table}[!h]
	\centering
	\caption{Comparisons among typical SQKD and SQDC protocols. }\label{tableCompare}
    \begin{tabular}{ || p{13em} | p{4em} |  p{3em}| p{5em}| p{5em}| l||}
    \hline
    Protocols (The protocols using entanglements are marked by *)& Qubit permutation & Basis Switch & Qubit Efficiency~$\eta$& Minimum \# quregisters&\gls*{ept} \\\hline\hline
    Boyer(2009) Randomization-Based SQKD \citep{PhysRevA.79.032341}	& Yes	& Yes	& $<12.5\%$	& $4n$	& $0$\\ \hline
    Boyer(2009) Measure-Resend SQKD \citep{PhysRevA.79.032341}		& No		& Yes	& $<12.5\%$	& $0$	& $0$\\ \hline
    Zou (2009) Protocol 5 \citep{PhysRevA.79.052312}					& No		& Yes	& $<12.5\%$	& $0$	& $0$\\ \hline
    Wang (2011) \citep{0256-307X-28-10-100301} *						& Yes	& Yes	& $<50\%$	& $6n$	& Worse than linear\\ \hline
    Li (2016) \citep{li2016} 										& No		& Yes 	& $<6.25\%$	& $0$	& $0$\\ \hline
    Luo (2016) \citep{Luo2016}	*									& Yes	& Yes	& $12.5\%$	& $20n$	& Worse than linear\\ \hline
    \gls*{nk} *														& No		& No		& $\approx 100\%$ 	& $1$	& $C+2T$\\ \hline
    \gls*{renk} *													& No		& No		& $\approx 99\%$		& $1$	& $C+2T$\\ \hline
    \end{tabular}
\end{table}

\end{document}